\documentclass[10pt,a4paper]{article}
\usepackage[a4paper]{geometry}                % See geometry.pdf to learn the layout options. There are lots.

\usepackage{cmbright}
\usepackage{bbm}
\usepackage[english]{babel}  %als het geïnstalleerd is; schrijft "hoofdstuk" enz.

\usepackage{amsfonts} %om bvb mathbb te kunnen gebruiken.
\usepackage{amsmath} % om functie in 2 lijnen te kunnen definiëren
\usepackage{amssymb}
\usepackage{graphics}
\usepackage{mathtools}
\usepackage{amsthm}
\usepackage{color}
\usepackage{xcolor}
\usepackage{xcolor}
\usepackage{amsthm}
\usepackage{framed}
\usepackage{enumerate} 
\usepackage{setspace}
\usepackage{stackrel}
\usepackage{layout}
\usepackage{tensor}
\usepackage{graphicx}
\usepackage{epstopdf}
\newtheorem{theorem}{Theorem}[section]
\newtheorem{main}[theorem]{Theorem (Main Result)}
\newtheorem{proposition}[theorem]{Proposition}

\newtheorem{definition}[theorem]{Definition}

\newtheorem{remark}[theorem]{Remark}

 \newcommand{\A}{\mathcal{A}} 
  \newcommand{\B}{\mathcal{B}}
    \newcommand{\C}{\mathcal{C}}
    \newcommand{\D}{\mathcal{D}}
\newcommand{\Q}{\mathcal{Q}}
        \newcommand{\Hi}{\mathcal{H}}

        \newcommand{\G}{{C(\mathbb{G})}}

        \newcommand{\OG}{{\mathcal{O}(\mathbb{G})}}

         \newcommand{\Ha}{{C(\mathbb{H}_1)}}
        \newcommand{\OHa}{{\mathcal{O}(\mathbb{H}_1)}}
                 \newcommand{\Hb}{{C(\mathbb{H}_2)}}
        \newcommand{\OHb}{{\mathcal{O}(\mathbb{H}_2)}}

     \newcommand{\Ga}{{C(\mathbb{G}_1)}}
        
                \newcommand{\OGa}{{\mathcal{O}(\mathbb{G}_1)}}

     \newcommand{\Gb}{{C(\mathbb{G}_2)}}
        
                \newcommand{\OGb}{\mathcal{O}(\mathbb{G}_2)}

  \newcommand{\Gdual}{\hat{\mathbb{G}}}

        \newcommand{\IrredG}{\Irred(\mathbb{G})}
                \newcommand{\IrredGa}{\Irred(\mathbb{G}_1)}

        \newcommand{\IrredHa}{\Irred(\mathbb{H}_1)}

        \newcommand{\vtimes}{\overline{\otimes}}

  \newcommand{\boksdot}[1][H]{\stackrel[#1]{}{\boxdot}}
  \newcommand{\bokstimes}[1][H]{\stackrel[#1]{}{\boxtimes}}

\let \DMO \DeclareMathOperator
\DMO{\Hom}{Hom}
\DMO{\Mor}{Mor}
\DMO{\Irred}{Irred}
\DMO{\dom}{dom}

\DMO{\Ker}{Ker}
\DMO{\GL}{GL}
\DMO{\SL}{SL}
\DMO{\M}{M}

\DMO{\vect}{vect} 
\DMO{\id}{id}
\DMO{\ad}{ad}
\DMO{\mult}{mult}
\DMO{\sgn}{sgn}

 \DMO{\inv}{inv}
 \DMO{\QISORnulb}{QISO_R^0}
 \DMO{\QISOb}{QISO}

              \newcommand{\QISOR}[1][\A,\Hi,D]{\QISOb_R(#1)}
         \newcommand{\QISORnul}[1][\A,\Hi,D]{\QISOb^0_R(#1)}
           \newcommand{\QISORnuleen}[1][\A_1,\Hi,D]{\QISOb^0_R(#1)}

    \newcommand{\QISORtilde}[1][\tilde \A,\tilde\Hi,\tilde D]{\QISOb_{\tilde R}(#1)}
         \newcommand{\QISORtildenul}[1][\tilde \A,\tilde \Hi,\tilde D]{\QISOb^0_{\tilde R}(#1)}

\DMO{\Tr}{Tr}
\title{Deformations of spectral triples and their quantum isometry groups via monoidal equivalences}
\author{Liebrecht De Sadeleer \\ Department of Mathematics, KU Leuven\footnote{Celestijnenlaan 200B, 3001 Leuven, Belgium, Email: liebrecht.desadeleer@wis.kuleuven.be}}
\date{}                                           % Activate to display a given date or no date
\numberwithin{equation}{section}

\begin{document}
\maketitle
\begin{abstract} 
In this paper, we propose a new procedure to deform spectral triples and their quantum isometry groups. The deformation data are a spectral triple $(\A,\Hi, D)$, a compact quantum group $\mathbb G$ acting algebraically and by orientation-preserving isometries on $(\A,\Hi,D)$ and a unitary fiber functor $\psi$ on $\mathbb G$. The deformation procedure is a proper generalization of the cocycle deformation of Goswami and Joardar. %Finally, we prove that the quantum isometry group of the deformed spectral triple is found through $\psi$.
\end{abstract}
\tableofcontents

\section*{Introduction}
An important source of examples of non-commutative manifolds in the sense of A. Connes (spectral triples, \cite{Connes1994}) relies on 2-cocycle deformations. For instance, the so-called `isospectral deformations' (\cite{Connes2001}) of compact spin manifolds admitting an action of a torus (or an action of the abelian group $\mathbb R^d$) may be seen as a by-product of Rieffel's machinery which, given a $C^*$- or Fr\'echet-algebra $\mathbb A$ on which $\mathbb R^d$ acts, produces a one-parameter continuous field of $C^*$-algebras $\{\mathbb A_\theta\}_{\theta \in \mathbb R}$ with $\mathbb A_0=\mathbb A$. The cocycle involved in this case is the usual Moyal 2-cocycle in $\mathbb R^d$. When $\mathbb A$ is the algebra underlying a spectral triple $(\mathbb A,\Hi,D)$, and the action of $\mathbb R^d$ lifts to an isometric action on $\Hi$, i.e. an action commuting with $D$, Rieffel machinery produces a new (family of) spectral triple(s) $(\mathbb A_\theta,\Hi,D)$. The paradigm there consists in the noncommutative torus within its metric version. \\ \\
In the present work, we generalize the deformation procedure through quantum group 2-cocycles (Goswami-Joardar, \cite{Goswami2014}) which is a way to produce new spectral triples from a given one. Our procedure is based on the notion of monoidal equivalence (introduced by Bichon , de Rijdt and Vaes, \cite{Bichon2005}) of (some subgroup of) its quantum isometry group (\cite{Goswami2009}). The generalized procedure here leads to examples that cannot be obtained by 2-cocycle deformations. 
%One obtains the following diagram:
%
%\begin{center}
%$\begin{diagram}
%(\A,\Hi,D)& \curvearrowleft & \mathbb{G}_1 \\ 
% \dDashto&&\dTo_\varphi \\ 
%(\A\boksdot[\OGa]\B,\Hi\bokstimes[\Ga] L^2(\B),(D\otimes 1_{L^2(\B)})_{|_{\Hi\bokstimes[\Ga] L^2(\B)}})&\qquad \curvearrowleft \qquad& \mathbb{G}_2
%\end{diagram}$
%\end{center} 
\\ \\
The paper is structured as follows. In the first section we recall some basic material and in the second, we describe the deformation procedure. In the third section we show that 2-cocycle deformations are particular cases of our deformation procedure. Moreover not all examples are coming from 2-cocycles: in the fourth section, we give such an example that is not a 2-cocycle deformation, proving our procedure is a proper generalization of the one by Goswami and Joardar. Finally in the last section, we prove that the quantum isometry group of the deformed spectral triple is a certain deformation of the quantum isometry group of the original spectral triple.
\\ \\
Before we end this introduction, we will clarify some notation. Given a Hilbert space $\Hi$, the inner product $\langle\cdot,\cdot \rangle$ is linear in the second variable. Moreover, for $\xi,\eta\in \Hi$, $\xi^*$ is the functional $\Hi\to \mathbb C:\eta\mapsto \langle \xi,\eta\rangle$ and $\xi\eta^*$ the rank one operator $\Hi\to \Hi:\zeta\to \xi\langle \eta,\zeta\rangle$. We will denote by $B(\Hi)$ resp. $\mathcal{K}(\Hi)$ the bounded resp. compact operators on $\Hi$ and for a bounded or unbounded operator $D$ on $\Hi$, $\sigma(D)$ will be used to denote its spectrum. Given a $C^*$-algebra $A$, the multiplier algebra of $A$ will be denoted by $\mathcal{M}(A)$ and for a subset $B$ of $A$, we define $\langle B \rangle$ to be the linear span of $B$, $[B]$ the closed linear span, $S(B)$ the $^*$-algebra generated by the elements of $B$ and $C^*(B)$ the $C^*$-subalgebra of $A$ generated by the elements of $B$. Furthermore, we use $\omega_{\xi,\eta}$ to denote the linear functional which maps $a\in B(\Hi)$ to $\langle \xi,a\eta\rangle$ where $\xi,\eta\in \Hi$, having linearity in the inner product in the second variable.
\\ \\
An algebraic tensor product will be denoted by $\odot$ while the minimal $C^*$-algebraic tensor product and a tensor product of Hilbert spaces is denoted by $\otimes$. %The tensor product of von Neumann algebras is denoted by $\vtimes$. 
We will also use the legnumbering notation in three and multiple tensor products: for $a\in A\otimes A$, we let $a_{12}=a\otimes 1_A, a_{23}=1_A\otimes a_{23}$ and $a_{13}=(\id\otimes \tau)(a\otimes 1_A)$, all three elements in $A\otimes A\otimes A$ where $\tau(a\otimes b)=b\otimes a$.
\\ \\
For a Hopf algebra $H$, the coproduct, counit and antipode will be denoted by $\Delta$, $\varepsilon$ and $S$ resp. We also use the Sweedler notation $\Delta(h)=h_{(1)}\otimes h_{(2)}$. A left, resp. right $H$-comodule is a vector space $A$ endowed with a linear map $\alpha: A\to H\odot A$ resp. $\alpha:A\to A\odot H$ satisfying $(\Delta\otimes \id)\alpha=(\id\otimes \alpha)\alpha$ resp. $(\alpha\otimes \id)\alpha=(\id\otimes \Delta)\alpha$. If $A$ is an algebra and $\alpha$ is multiplicative, it is called a coaction of $H$ on $A$ and $A$ is called an $H$-comodule algebra. If $A$ and $B$ are a right resp. left $H$-comodule algebra with resp. coactions $\alpha$ and $\beta$, $A\boksdot B$ will denote the algebra $\{z\in A\odot B| (\alpha\otimes \id)(z)=(\id\otimes \beta)(z)\}$.

\section{Compact quantum groups and Monoidal equivalences}\label{sec-cqg}

We start this section with a short overview of the theory of compact quantum groups. The theory is essentially developed in \cite{Woronowicz1987}, \cite{Woronowicz1998} and also explained in \cite{Maes1998}.

\subsection{Compact quantum groups and representations}\label{subsec-cqg}

\begin{definition}[\cite{Woronowicz1998}]
A compact quantum group $\mathbb{G}$ is a pair $(\G,\Delta)$ where $\G$ is a unital, seperable $C^*$-algebra and $\Delta : \G\to \G\otimes \G$ a unital $^*$-morphism such that
\begin{enumerate}
\item $(\Delta\otimes \id)\Delta = (\id\otimes \Delta)\Delta$
\item $[\Delta(\G)(\G\otimes 1)]=\G\otimes \G=[\Delta(\G)(1\otimes \G)]$
\end{enumerate}
\end{definition}
implementing coassociativity and the cancellation properties.

Moreover there exists a unique state $h$ on $\G$ which is left and right invariant in the sense that $(\id\otimes h)\Delta(x)=h(x)1_\G=(h\otimes \id)\Delta(x)$ for all $x\in \G$ (\cite{Woronowicz1987,Woronowicz1998,Maes1998}).This state is called the Haar state of $\mathbb{G}$.
In the classical case that $\G=C(G)$ for a classical compact group $G$, the Haar state is the state on $C(G)$ obtained by integrating along the Haar measure.\\
\\
It is well known that, like compact groups, compact quantum groups have a rich representation theory (\cite{Woronowicz1987,Woronowicz1998,Maes1998}). A right unitary representation of a compact quantum group $\mathbb{G}=(\G,\Delta)$ on a Hilbert space $\Hi$ is a unitary element $U$ of $\mathcal{M}(\mathcal{K}(\Hi)\otimes \G)$ satisfying
$ (\id \otimes \Delta)U=U_{12}U_{13}.$ Analgously, a left unitary representation of $\mathbb G$ on $\Hi$ is a unitary element $U$ of $\mathcal{M}(\G\otimes \mathcal{K}(\Hi))$ satisfying
$ (\Delta \otimes\id)U=U_{13}U_{23}.$ In this paper all representations will be right representations unless indicated otherwise.
The dimension of $\Hi$ is called the dimension of the representation.
Identifying $\mathcal{M}(\mathcal{K}(\Hi)\otimes \G)$ with $B(\Hi\otimes \G)$, the $C^*$-algebra of $\G$-linear adjointable maps on the Hilbert-$C^*$-module $\Hi\otimes \G$, we will also see representations as maps $u:\Hi\to\Hi\otimes \G:\xi \to U(\xi\otimes 1_{\G})$ satisfying that $\langle u(\xi),u(\eta)\rangle_{\G}=\langle \xi,\eta\rangle 1_{\G}$, $(u\otimes \id)u=(\id\otimes \Delta)u$ and $[u(\xi)(1\otimes a):\xi\in \Hi,a \in \G]=\Hi\otimes \G$.

Moreover, there is the notion of tensor product of representations: if $U$ and $V$ are representations of a quantum group $\mathbb{G}=(\G,\Delta)$ on Hilbert spaces $\Hi_1,\Hi_2$ respectively, the tensor product $U\otimes V$ of $U$ and $V$ is defined as $U\otimes V = U_{13}V_{23}\in \mathcal{M}(\mathcal{K}(\Hi_1\otimes\Hi_2)\otimes \G)$. Furthermore, we call a representation $U$ of $\mathbb G$ on $\Hi$ irreducible if $\Mor(U,U)= \mathbb C 1_{B(\Hi)}$ where $$\Mor(U^1,U^2):=\{S\in B(\Hi_2,\Hi_1) | (S\otimes 1_\G)U^2=U^1(S\otimes 1_\G)\}$$ for representations $U_1$ and $U_2$ on $\Hi_1$ resp. $\Hi_2$. An important result states that every irreducible representation is finite dimensional and that every unitary representation is unitarily equivalent to a direct sum of finite dimensional irreducible representations. Finally, for every irreducible unitary representations, there exist the notion of contragredient representation (\cite{Woronowicz1998},\cite{Maes1998}). \\

For a compact quantum group $\mathbb{G}$, we denote by $\Irred(\mathbb{G})$ the set of equivalence classes of irreducible representations of $\mathbb{G}$ and for $x\in\Irred(\mathbb{G})$, we will always take a unitary representative $U^x\in B(\Hi_x)\otimes \G$. By $\varepsilon$, we will denote the class of the trivial representation $1_\G$.

Also for a compact quantum group $\mathbb{G}=(\G,\Delta)$ and an equivalence class $x\in \Irred(\mathbb G)$, we will denote by $(\omega_{\xi,\eta}\otimes \id_\G)U^x$ a matrix coefficient where $\xi,\eta \in \Hi$ and define $\OG$ to be the linear span of matrix coefficients of all irreducible (hence finite dimensional) representations of $\mathbb{G}$:
$$\OG=\langle (\omega_{\xi,\eta}\otimes \id_\G)U^x| x\in \Irred(\mathbb{G}), \xi,\eta \in \Hi_x \rangle,$$ even more, the matrix coefficients of the irreducible representations form a basis of $\OG$. Note that $\OG$ is a unital dense $^*$-subalgebra of $\G$ which has, endowed with the restriction of $\Delta$ to $\OG$, the structure of a Hopf $^*$-algebra. This is a very nontrivial result obtained in \cite{Woronowicz1998}, see also \cite{Maes1998}. Also, for a $x\in \Irred(\mathbb{G})$, let $\OG_x=\langle (\omega_{\xi,\eta}\otimes \id_\G)U^x| \xi,\eta \in \Hi_x \rangle$. Then we have $\Delta:\OG_x\to \OG_x\odot\OG_x$ and $\OG_x^*=\OG_{\overline{x}}$.

\begin{definition}[\cite{Bedos2001}]
Let $\mathbb{G}$ be a compact quantum group. The reduced $C^*$-algebra $C_r(\mathbb{G})$ is defined as the norm closure of $\OG$ in the GNS-representation with respect to the Haar state $h$ of $\mathbb{G}$.The universal $C^*$-algebra $C_u(\mathbb{G})$ is defined as the $C^*$-envelope of $\OG$. 
%Finally, the von Neumann algebra $L^\infty(\mathbb{G})$ is defined as the von Neumann algebra generated by $C_r(\mathbb{G})$. 
Note that if $\mathbb{G}$ is the dual of a discrete (classical) group $\Gamma$, we have $C_r(\mathbb{G})=C^*_r(\Gamma), C_u(\mathbb{G})=C^*_u(\Gamma)$.
% and $L^\infty(\mathbb{G})=\mathcal{L}(\Gamma)$
\end{definition}

\begin{remark}\label{remarkversions}
Note that, for a given compact quantum group $\mathbb{G}$, we have surjective morphisms between the different completions of $\OG$: $C_u(\mathbb{G})\to C(\mathbb{G})\to C_r(\mathbb{G})$. We will think of all these algebras as describing the same quantum group.
\end{remark}

\begin{definition}
Let $\mathbb{G}=(C_u(\mathbb G),\Delta_{\mathbb{G}})$ and $\mathbb{H}=(C_u(\mathbb{H}),\Delta_{\mathbb{H}})$ be compact quantum groups equipped with their universal $C^*$-norms. Suppose moreover that there exists a surjective map $\theta:\G\to C(\mathbb{H})$ satisfying $\Delta_{\mathbb{H}}\circ \theta=(\theta\otimes\theta) \Delta_{\mathbb{G}}$. Then we call $\mathbb{H}$ a quantum subgroup of $\mathbb{G}$. Equivalently, $\mathbb G$ is called a quantum supergroup of $\mathbb H$.
\end{definition}

%\subsection{Discrete quantum groups and duals of compact quantum groups}
\begin{definition}[\cite{Maes1998}]
Let $\mathbb G$ be a compact quantum group. Let 
$$c_0( \hat{\mathbb G})= \oplus_{x\in \IrredG}B(\Hi_x),\qquad \ell^\infty(\Gdual)=\prod_{x\in\IrredG}B(\Hi_x).$$ Then we call $\Gdual$ the dual quantum group which has the structure of a discrete quantum group (see \cite{VanDaele1996} for the definition and results). 
\end{definition}
Using the notation $\mathbb V=\oplus_{x\in \IrredG}U^x$, we can define the dual comultiplication
$$\hat \Delta: \ell^\infty(\Gdual)\to \ell^\infty(\Gdual)\vtimes\ell^\infty(\Gdual) : (\hat \Delta \otimes \id)(\mathbb V)=\mathbb V_{13} \mathbb{V}_{23}.$$

\subsection{Actions of compact quantum groups and the spectral subalgebra}\label{subsec-actions}
\begin{definition}[\cite{Podles1995}]
Let $B$ be a unital $C^*$-algebra and $\mathbb{G}=(\G,\Delta)$ a compact quantum group. A right action of $\mathbb{G}$ on $B$ is a unital $*$-homomorphism $\beta: B\to B\otimes \G$ such that 
\begin{enumerate}
\item $(\beta\otimes \id_\G)\beta=(\id_B\otimes \Delta)\beta$
\item $[\beta(B)(1\otimes \G)]=B\otimes \G$.
\end{enumerate}
Analogously, a left action is a unital $^*$-morphism $\beta': B\to \G \otimes  B$ satisfying the analogous conditions. 
We say that the action is ergodic if $B^\beta=\{b\in B|\beta(b)=b\otimes 1\}=\mathbb{C}1_B$.
\end{definition}

One can choose to call the map in this definition `a coaction' as it is a coaction of the $C^*$-algebra $\G$ on $B$. However, we choose to call it an action of the compact quantum group in order to be compatible with the classical case: if $\G=C(G)$ and $B=C(X)$ with $G$ a classical compact group and $X$ a compact space, it is an action of $G$ on $X$. \\
One can prove that in the case of ergodic actions, there is a unique invariant state on $B$ (\cite{Boca1995}), which we will denote by $\omega$.

Note that the most evident example is a quantum group acting on itself by comultiplication. In that situation, one can check that $\omega=h$.

Using the intimate link between the ergodic action of a compact quantum group on a unital $C^*$-algebra and the representations of the quantum group, one has the following result.

\begin{proposition}[\cite{Boca1995}]
Let $B$ be a unital $C^*$-algebra and $\beta:B\to \G\otimes B$ a left action of $\mathbb{G}$ on $B$. Define for every $x\in\IrredG$, 
$$K_x=\{\zeta\in \Hi_x\otimes B\mid U^x_{12}\zeta_{13}=(\id_{\Hi_x}\otimes \beta)\zeta\}$$ and
$$\B_x=\langle (\xi^*\otimes 1_B)\zeta\mid \zeta\in K_x,\xi\in \Hi_x\rangle.$$
Then the spaces $\B_x$ with $x\in \IrredG$ are called the spectral subspaces of $B$ and $$\B=\langle (\xi^*\otimes 1_B)\zeta\mid x\in\IrredG,\zeta\in K_x,\xi\in \Hi_x\rangle$$ is a dense unital $^*$-subalgebra of $B$ which we will call the spectral subalgebra of $B$ with respect to $\beta$. Moreover $\beta_{|_{\B}}$ is an algebraic coaction of the Hopf $^*$-algebra $(\OG,\Delta)$ on $\B$. 
\end{proposition}

\begin{remark}
An action $\beta:B\to B\otimes \G$ of $\mathbb{G}$ on $B$ is called universal if $B$ is the universal $C^*$-algebra of $\B$. It is called reduced if the map $(\id\otimes h)\beta:B\to B$ onto the fixed point algebra $B^\beta$ is faithful.\\
In remark \ref{remarkversions} we saw that a compact quantum group can be described using different $C^*$-algebras, having the same underlying (dense) Hopf $^*$-subalgebra. Similarly here, given an action $\beta :B\to B\otimes \G$ of $\mathbb{G}$ on $B$, passing through $\B$ we can associate to it its universal and reduced $C^*$-completions $B_u$ and $B_r$, and we have surjective morphisms: $B_u\to B\to B_r$. %So again, we will identify two actions if the underlying Hopf $^*$-algebras are the same.
\end{remark}

\subsection{Monoidal equivalences between compact quantum groups}\label{subsec-moneq}
\begin{definition}[\cite{Bichon2005}]\label{defmoneq}
Let $\mathbb{G}_1=(\Ga,\Delta_1)$ and $\mathbb{G}_2=(\Gb,\Delta_2)$ be two compact quantum groups. $\mathbb{G}_1$ and $\mathbb{G}_2$ are called monoidally equivalent if there exists a bijection $\varphi: \Irred(\mathbb{G}_1)\to \Irred(\mathbb{G}_2)$ which satisfies $\varphi(\varepsilon_{\mathbb{G}_1})=\varepsilon_{\mathbb{G}_2}$ together with linear isomorphisms:
\begin{multline*}
\varphi:\Mor(x_1\otimes\ldots\otimes x_r,y_1\otimes\ldots\otimes y_k)\\\to \Mor(\varphi(x_1)\otimes\ldots\otimes \varphi(x_r),\varphi(y_1)\otimes\ldots\otimes \varphi(y_k))\qquad\qquad
\end{multline*}
satisfying 
\begin{eqnarray}\label{condmoneq}
&\varphi(1)=1,\qquad \varphi(S\otimes T)=\varphi(S)\otimes \varphi(T),&\nonumber\\
&\varphi(S^*)=\varphi(S)^*,\qquad \varphi(ST)=\varphi(S)\varphi(T)&
\end{eqnarray}
whenever the formulas make sense. The collection of maps is called a monoidal equivalence.
\end{definition}
Note that this is indeed the usual definition of equivalence between strict monoidal categories, but adapted to the concrete case of the category of representations of a compact quantum group.
\begin{definition}[\cite{Bichon2005}]\label{defunitfiberfunctor}
Let $\mathbb G= (\G,\Delta)$ be a compact quantum group. A unitary fiber functor is a collection of maps $\psi$ such that 
\begin{itemize}
\item for every $x\in \IrredG$, there is a finite dimensional Hilbert space $\Hi_{\psi(x)}$,
\item there are linear maps 
\begin{multline}
\psi: \Mor(x_1\otimes \ldots\otimes x_k,y_1\otimes \ldots\otimes y_s)\\\to B(\Hi_{\psi(y_1)}\otimes \ldots\otimes\Hi_{\psi(y_s)}, \Hi_{\psi(x_1)}\otimes \ldots \otimes \Hi_{\psi(x_k)})
\end{multline} 
which satisfy the equations \eqref{condmoneq} of definition \ref{defmoneq}.
\end{itemize}
\end{definition}
\begin{remark}[\cite{Bichon2005}]\label{unitfibfunctinducesmoneq}
To define a unitary fiber functor it suffices to attach to every $x\in \IrredG$ a finite dimensional Hilbert space $\Hi_{\psi(x)} (\Hi_\varepsilon=\mathbb C)$ and to define the linear maps
$$\psi: \Mor(x_1\otimes \ldots\otimes x_k,y)\to B(\Hi_{\psi(y)},\Hi_{\psi(x_1)}\otimes \ldots\otimes \Hi_{\psi(x_k)})$$ for $k=1,2,3$ satsifying 
\begin{eqnarray}
\psi(1)=1&&\\
\psi(S)^*\psi(T)=\psi(S^*T) & \text{if} &S\in \Mor(x\otimes y,a), T\in \Mor(x\otimes y,b)\\
(\psi(S)\otimes \id)\psi(T)=\psi((S\otimes \id)T) & \text{if} &S\in \Mor(x\otimes y,a), T\in \Mor(a\otimes z,b)\\
(\id\otimes\psi(S) )\psi(T)=\psi((\id\otimes S )T) & \text{if} &S\in \Mor(x\otimes y,a), T\in \Mor(a\otimes z,b)
\end{eqnarray}
together with a non-degenerateness condition 
\[[\psi(S)\xi|a\in \IrredG,S\in \Mor(b\otimes c,a),\xi\in \Hi_{\psi(a)}]=\Hi_{\psi(b)}\otimes \Hi_{\psi(c)}\]
\end{remark}

In fact, the notions of unitary fiber functor and monoidal equivalence are equivalent, which is stated in the following proposition, taken from Proposition 3.12 in \cite{Bichon2005}.
\begin{proposition}\label{unfibfunctismoneq}
Let $\mathbb G_1$ be a compact quantum group and $\psi$ a unitary fiber functor on it. Then there exist a unique universal compact quantum group $\mathbb{G}_2$ with underlying Hopf algebra $(\OGb,\Delta_2)$ with unitary representations $U^{\psi(x)}\in B(\Hi_{\psi(x)})\otimes \Gb$, $x\in \Irred(\mathbb G_1)$ such that
\begin{enumerate}
\item $U^{\psi(y)}_{13}U^{\psi(z)}_{23}(\psi(S)\otimes 1)=(\psi(S)\otimes 1)U^{\psi(x)}$ for all $S\in \Mor(y\otimes z,x)$,
\item the matrix coefficients of the $U^{\psi(x)}$, $x\in \Irred(\mathbb G_1)$ form a linear basis of $\OGb$.
\end{enumerate}
Moreover, the set $\{U^{\psi(x)}|x\in \mathbb{G}_1\}$ forms a complete set of irreducible representations of $\mathbb G_2$ and the unitary fiber functor $\psi$ on $\mathbb G_1$ will induce a monoidal equivalence $\varphi:\mathbb{G}_1\to \mathbb{G}_2$.
\end{proposition}

The following theorems of Bichon et al. will be crucial in our main result. They explain what extra structure a monoidal equivalence induces.

The first theorem follows from Theorem 3.9 and Proposition 3.13 of \cite{Bichon2005}.
\begin{theorem}[\cite{Bichon2005}]\footnote{In the original statement of \cite{Bichon2005}, the coaction $\beta_1$ is a right coaction of $\OGa$, but for what follows, we want a left coaction of $\OGa$ and a right coaction of $\OGb$. Applying Bichon's theorem on the inverse monoidal equivalence $\varphi':\mathbb G_2\to \mathbb G_1$, one gets the theorem stated here. Note that, when doing that, we should write $X^{\varphi(x)}$, $x\in \Irred(\mathbb G_1)$ but for notational convenience, we write $X^{x}, x\in\Irred(\mathbb G_1).$}\label{monoidalGalois}
Let  $\mathbb{G}_1$ be a compact quantum group and let $\psi$ be a unitary fiber functor on $\mathbb G_1$. Denote with $\varphi :\mathbb G_1 \to \mathbb G_2 $ the monoidal equivalence induced by $\psi$ (see previous proposition).
 \begin{enumerate}
\item There exists a unique unital $^*$-algebra $\B$ equipped with a faithful state $\omega$ and unitary elements $X^x \in B(\Hi_{\varphi(x)},\Hi_{x})\odot \B$ for all $x \in \Irred(\mathbb{G}_1)$ satisfying
\begin{enumerate}
\item $X^y_{13}X^z_{23}(\varphi(S)\otimes 1)=(S\otimes 1)X^x$ for all $S\in \Mor(y\otimes z,x)$,
\item the matrix coefficients of the $X^x$ form a linear basis of $\B$,
\item $(\id\otimes \omega)(X^x)=0$ if $x\neq \varepsilon$.
\end{enumerate}

\item There exist unique commuting coactions $\beta_1 : \B \to \OGa\odot \B $ and $\beta_2 : \B \to \B \odot \OGb$ satisfying $$ (\id\otimes \beta_1)(X^x)=U^{x}_{12}X^x_{13}\qquad \text{and} \qquad (\id\otimes \beta_2)(X^x)=X^x_{12}U^{\varphi(x)}_{13} $$ for all $x\in \IrredG$. Moreover, $\omega(b)1_B=(h\otimes\id_B)\beta_1(b)$.
\item The state $\omega$ is invariant under $\beta_1$ and $\beta_2$. Denoting by $B_{r}$ the $C^*$-algebra generated by $\B$ in the GNS-representation associated with $\omega$ and denoting by $B_u$ the universal enveloping $C^*$-algebra of $\B$, the Hopf algebraic coactions $\beta_1$ and $\beta_2$ admit unique extensions to actions of the compact quantum groups on $B_r$, resp. $B_u$. These actions are reduced, resp. universal and they are ergodic and of full quantum multiplicity (see \cite{Bichon2005} for the definition).
%\item Every reduced, resp. universal, ergodic action of full quantum multiplicity, arises in this way from a monoidal equivalence.
\end{enumerate}
\end{theorem}

\begin{definition}
In what follows, we will call $\B$ the $\mathbb G_1-\mathbb G_2$-bi-Galois object associated with $\varphi$. 
\end{definition}

In the spirit of this theorem, we can introduce the notion of isomorphism of unitary fiber functors, which will be equivalent to the isomorphism of the associated bi-Galois objects.

\begin{definition}[Def. 3.10 in \cite{Bichon2005}]
Let $\psi$ and $\psi'$ be two unitary fiber functors on a compact quantum group $\mathbb G$. We say they are isomorphic if there exist unitaries $u_x\in B(\Hi_{\varphi(x)},\Hi_{\psi(x)})$ such that
\[\psi'(S)=(u_{y_1}\otimes\ldots\otimes u_{y_k})\psi(S)(u^*_{x_1}\otimes \ldots\otimes u^*_{x_r})\] for all $S\in\Mor(y_1\otimes\ldots\otimes y_k,x_1\otimes \ldots\otimes x_r)$.

\begin{proposition}
Let $\psi$ and $\psi'$ be two unitary fiber functors on a compact quantum group $\mathbb G$. Let $\B_{\psi'}$ and $\B_\psi$ be the associated bi-Galois objects with respective coactions $\beta_{\psi}, \beta_\psi'$. Then $\psi$ and $\psi'$ are isomorphic as unitary fiber functors if and only if there exists a $^*$-isomorphism $\lambda: \B_\psi\to \B_{\psi'}$ satisfying $(\lambda\otimes \id)\beta_\psi=\beta_{\psi'}\lambda$.
\end{proposition}

%\begin{remark}\label{rem-characterizationOGfromB}
%One can easily see that if $\B$ is a ($\mathbb G_1$-$\mathbb G_2$)-bi-Galois object, then $$\OGa=S(\{(\id_{\OGa}\otimes f_B)\beta_1(b)\mid b\in \B,f_B\in \B^*\})$$ and $$\OGb=S( \{(f_B\otimes\id_{\OGb})\beta_2(b)\mid b\in \B,f_B\in \B^*\}).$$ We will use this characterization in the last section.
%\end{remark}

\end{definition}
There is even more, De Rijdt and Vander Vennet proved in \cite{DeRijdt2010} that there exists a bijection between actions of monoidal equivalent compact quantum groups. Indeed, let $\mathbb{G}_1$ and $\mathbb{G}_2$ be two compact quantum groups, $\varphi: \mathbb{G}_1\to  \mathbb{G}_2$ be a monoidal equivalence between them. Let $\B,\beta_1,\beta_2,X^x$ be as in the previous theorem. Suppose moreover that we have a $C^*$-algebra $D_1$ and an action $\alpha_1: D_1\to D_1\otimes \Ga$ of $\mathbb G_1$ on $D_1$. Using the dense Hopf $^*$-algebras, we have a coaction $\alpha_1:\D_1\to \D_1\odot \OGa$ of $\OGa$ on $\D_1$ and we can define the $^*$-algebra:
$$ \D_2=\D_1\stackrel[\OGa]{}{\boxdot}\B=\{a\in \D_1\odot \B| (\alpha_1\otimes \id_\B)(a)=(\id_{\D_1}\otimes \beta_1)(a)\}.$$ Moreover, in \cite{DeRijdt2010}, the authors prove that  the same construction with the inverse monoidal equivalence $\varphi^{-1}$ will give $\D_1$ again up to isomorphism.
%and analogously on the von Neumann algebraic level, we have:
%$$ D_2^{vN}=D_1^{vN}\stackrel[L^\infty(\mathbb{G})]{}{\overline{\boxtimes}}B^{vN}=\{a\in D\overline{\otimes} B^{vN}| (\alpha_1\otimes \id_{B^{vN}})(a)=(\id_{D_1^{vN}}\otimes \beta_1)(a)\}.$$

\begin{theorem}\label{actiondeformation}
Given the data above, there exists an action $\alpha_2=(\id\otimes \beta_2)_{|_{\D_2}}$ on $\D_2$. Moreover, if $\alpha_1$ is ergodic, $\alpha_2$ is ergodic as well. % and $\mult_q(x)=\mult_q(\varphi(x))$ for all $x\in \Irred(\mathbb{G}_2)$. Finally, $\D_2$ is the spectral subalgebra of $D_2^{vN}$ and hence SOT-dense in $D_2^{vN}$.
\end{theorem}
To end this subsection, we have a look at the inverse monoidal equivalence. We rephrase Proposition 7.6 from \cite{DeRijdt2010} in our notations.

\begin{proposition}\label{propGbtildeBB}
Let $\mathbb G_1$ and $\mathbb G_2$ be two compact quantum group and $\varphi:\mathbb G_1\to \mathbb G_2$ a monoidal equivalence with bi-Galois object $\B$. Denote by $\varphi^{-1}:\mathbb G_2\to \mathbb G_1$ the inverse monoidal equivalence with bi-Galois object $\tilde \B$ generated by the matrix coefficients of unitaries $Z^y\in B(\Hi_{\varphi^{-1}(y)}, \Hi_{y})\odot \tilde \B$, $y\in \Irred(\mathbb G_2)$ and coactions $\delta_1:\tilde B\to \tilde B\odot \OGa$ and $ \delta_2:\tilde B\to \OGb\odot \tilde \B$ such that
$$(\id\otimes\delta_1)Z^y= Z^y_{12}U^{\varphi^{-1}(y)}_{13} \qquad \text{and} \qquad  (\id\otimes \delta_2)Z^y=U^y_{12} Z^y_{13}.$$
Then
$$\pi: \OGa\to \B\boksdot[\OGb]\tilde{\B} \qquad \text{with} \qquad (\id\otimes\pi)(U^x)=X^x_{12}Z_{13}^{\varphi(x)}$$ is a $^*$-isomorphism intertwining the comultiplication $\Delta_1$ with the coaction $(\beta_1\otimes \id)=(\id\otimes \delta_1)$.
\end{proposition}

\section{Deformation procedure for spectral triples}\label{sec-mondef}

Before we start with the description of the deformation procedure, we recapitulate the notion of spectral triples and that of CQG acting on spectral triples.
\subsection{Spectral triples and compact quantum groups acting on them}

\begin{definition}[\cite{Connes1994}]
A (compact) spectral triple $(\A, \Hi,D)$ consists of
\begin{enumerate}
\item a unital $^*$-algebra $\A$ acting as bounded operators on $\Hi$, 
\item a Hilbert space $\Hi$,
\item an unbounded selfadjoint operator $D$ on $\Hi$ with compact resolvent such that $[D,a]$ is bounded for all $a\in \A$.
% and $(1+D^2)^{-1}$ is compact.
\end{enumerate}
\end{definition}
%\begin{remark}
%Note that, denoting with $A$ the $C^*$-completion of $\A^\infty$, we have $$\A^\infty=\{a\in A |a\cdot\dom(D)\subseteq \dom(D) \text{ and } \| [D,a]\|<\infty\}.$$
%\end{remark}
%
\begin{definition}[\cite{Connes1994}]\label{isospectrtrpl}
Two spectral triples $(\A_1, \Hi_1,D_1)$ and $(\A_2, \Hi_2,D_2)$ are called isomorphic, if there exists an isomorphism of Hilbert spaces $\phi: \Hi_1\to \Hi_2$ and an isomorphism of $^*$-algebras $\lambda:\A_1\to \A_2$ such that $\phi D_1=D_2 \phi$ and $\phi(a\xi)=\lambda(a)\phi(\xi)$ for arbitrary $\xi \in \Hi_1, a\in \A_1$.
\end{definition}

In \cite{Bhowmick2009,Goswami2009} Bhowmick and Goswami described how compact quantum groups can act isometrically and orientation-preserving on a non-commutative manifold, i.e. a spectral triple.
\begin{definition}[\cite{Bhowmick2009}]
Let $(\A,\Hi,D)$ be a compact spectral triple, $\mathbb G=(\G,\Delta)$ a compact quantum group and $U$ a unitary representation of $\mathbb G$ on $\Hi$. Then $\mathbb G$ is said to act by orientation-preserving isometries on $(\A,\Hi,D)$ with $U$ if
\begin{itemize}
\item for every state $\phi$ on $M$, we have $U_\phi D=DU_\phi$ where $U_\phi := (\id \otimes \phi)(U)$,
\item $(\id \otimes \phi)\circ \alpha_U(a)\in \A''$ for all $a\in \A$ and state $\phi$ on $M$; where $\alpha_U(T):= U(T\otimes 1)U^*$ for $T\in B(\Hi)$.
\end{itemize}
\end{definition}

This definition is a very strong one: it ensures the existence of a universal object in the category of all compact quantum groups acting by orientation-preserving isometries. However, in some cases the second condition is to weak: the quantum group representation on $\Hi$ may behave badly with respect to the algebra $\A$ in the sense that the induced action of the CQG on $\A$ is not a CQG-action on the $C^*$-closure of $\A$. This is in some situations a disadvantage. Therefore, we note the following proposition of Goswami, found in \cite{Goswami2014}.

\begin{proposition}\label{prop-existsalgactsubalg}
Let $(\A,\Hi,D)$ and $(\G,\Delta,U)$ be as above. Then there exists an algebra $\A_1$ such that 
\begin{enumerate}
\item $\A_1$ is SOT-dense in the von Neumann Algebra $M=\A''$,
\item $\alpha_U$ is algebraic on $\A_1$, i.e. $(\alpha_U)_{|_{\A_1}}:\A_1\to \A_1\odot \OG$,
\item $[D,a]$ is bounded for every $a\in \A_1$,
\item $(\A_1,\Hi,D)$ is again a spectral triple.
\end{enumerate}
\end{proposition} 
\begin{proof}
This follows from sections 4.4.3 and 4.4.4 and theorem 4.10 in \cite{Goswami2014}
\end{proof}

Driven by proposition \ref{prop-existsalgactsubalg}, we will use the following definition:
\begin{definition}\label{def-algactcqg}
Let $(\A,\Hi,D)$ be a compact spectral triple, $\mathbb G=(\G,\Delta)$ a compact quantum group and $U$ a unitary representation of $\mathbb G$ on $\Hi$. Then $\mathbb G$ is said to act \textbf{algebraically} and by orientation-preserving isometries on $(\A,\Hi,D)$ with $U$ if
\begin{itemize}
\item for every state $\phi$ on $M$, we have $U_\phi D=DU_\phi$ where $U_\phi := (\id \otimes \phi)(U)$,
\item $\alpha_U$ is algebraic on $\A$, i.e. $(\alpha_U)_{|_{\A}}:\A\to \A\odot \OG$ where $\alpha_U(T):= U(T\otimes 1)U^*$ for $T\in B(\Hi)$.
\end{itemize}
\end{definition}

In what follows, we will always work with compact quantum groups acting algebraically on the algebra $\A$.

\subsection{Deformation procedure for spectral triples}
In this subsection we will describe the actual deformation procedure for spectral triples. The deformation data to start with are:
\begin{itemize}
\item a spectral triple $(\A,\Hi,D)$ of compact type,
\item a compact quantum group $\mathbb G_1=(\Ga,\Delta_1)$ acting algebraically and by orientation-preserving isometries on $(\A,\Hi,D)$ with a unitary representation $U$ and
\item  a unitary fiber functor $\psi$ on $\mathbb G_1$.
\end{itemize}

The unitary fiber functor will induce a new compact quantum group $\mathbb G_2$ and a $^*$-algebra $\B$ with left resp. right coaction of $\OGa$ resp. $\OGb$. Using this, one can deform the data one by one to obtain a new, deformed, spectral triple on which $\mathbb G_2$ acts in an appropriate way.

To be more precise, consider the following:

\begin{enumerate}
\item As $\psi$ is a unitary fiber functor on $\mathbb{G}_1$, following theorem \ref{monoidalGalois} there exists a compact quantum group $\mathbb G_2$ and a monoidal equivalence $\varphi:\mathbb G_1 \to \mathbb G_2$. We will call $\mathbb G_2$ the deformed quantum group.
\item Let $(\B,\omega)$ be the $^*$-algebra and faithful invariant state associated to $\varphi$ with the coactions 
$$\beta_1:\B\to \OGa\odot \B \qquad \text{and} \qquad \beta_2:\B\to \B\odot \OGb.$$ 
%and denote by $B_1^{red}$ the $C^*$-algebra generated by $\B$ in the GNS-representation associated to $\omega$. Let $\beta_1^{red}:B_1^{red}\to \Ga\otimes B_1^{red}$ be the action on the reduced $C^*$-algebra level. Moreover, let $B_1^{vN}$ be the von Neumann algebra generated by $(\B,\omega)$ and $\beta_1^{vN}:B_1^{vN}\to \LGa\vtimes B_1^{vN}$ the von Neumann algebraic action;
%\item

% and analogously, define $B_2^{red}$, $B_2^{vN}$ and $\beta_2^{red},\beta_2^{vN}$;
\item Let $X^x\in B(\Hi_{\varphi(x)},\Hi_x)\odot \B$ be the unitaries such that $$(\id\otimes \beta_1)X^x=U_{12}^{x}X^x_{13} \qquad \text{and} \qquad (\id\otimes \beta_2)X^x=X^x_{12}U_{13}^{\varphi(x)}.$$
\item Let $u: \Hi\to\Hi \otimes \Ga:\xi\mapsto U(\xi\otimes 1)$ be the representation of $\mathbb{G}_1$ on $\Hi$ and denote by $\alpha=\ad_U:\A\to \A\odot \OGa:a\to U(a\otimes1_\Ga)U^*$ the algebraic coaction of $\OGa$ on $\A$.
%\item let $M=\A''$ be the von Neumann algebra generated by $\A$ and $\alpha'': M\to M\vtimes \LG$ the von Neumann algebraic action.  
%Denote with $\A_0$ the spectral subalgebra of $M$ for this action.
\end{enumerate}

We start by introducing the deformed data and proving some basic facts about them.

\begin{proposition}\label{defhilbertspace}
Defining $L^2(\B)$ to be the GNS representation of $\B$ with respect to $\omega$ and $\Lambda: \B\to L^2(\B)$ the GNS map, we have:
\begin{enumerate}
\item  there exists a unitary representation $\beta'_1$ of $\Ga$ on $L^2(\B)$ such that $\beta'_1(\Lambda(b))=(\id\otimes \Lambda)(\beta_1(b))$. 
\item $\beta'_1$ is ergodic, i.e. if $\xi \in L^2(\B)$ such that $\beta_1'(\xi)=1\otimes \xi$, then $\xi \in \mathbb{C}\Lambda(1_\B)$.
\item The vector space $\Hi\stackrel[\Ga]{}{\boxtimes}L^2(\B)=\{\xi \in \Hi\otimes L^2(\B)| U_{12}\xi_{13}=(\id_\Hi\otimes \beta'_1)(\xi)\}$ is a Hilbert space which we denote by $\tilde{\Hi}$.
\end{enumerate}
\end{proposition}

\begin{proof}
\begin{enumerate}
\item As $\omega$ is faithful on $\B$, $\Lambda$ is injective and hence $\beta_1'$ is well defined on $\Lambda(\B)$. Using that $\beta_1$ is a well defined coaction of $\OGa$ and that $\omega$ is $\beta_1$-invariant, $\beta_1'$ can be extended to a unitary representation on $L^2(\B)$.

\item Let $\xi$ be an element in $L^2(\B)$ satisfying $\beta'_1(\xi)=1\otimes \xi$. Take a sequence $(b_n)_n$ in $\B$ with $\Lambda(b_n)\to \xi$ in $L^2$-norm, then for $P=(h\otimes \id_{L^2(\B)})\beta_1'$ we see that $P(\Lambda(b_n))\to P(\xi)=\xi$ since $P$ is a continuous operator on $L^2(\B)$. Seeing that $P(\Lambda(b_n))=\omega(b_n)\Lambda(1_\B)\in \mathbb C 1_\B$ concludes this proof.

\item It follows directly that $\Hi\stackrel[\Ga]{}{\boxtimes}L^2(\B)$ is a vector subspace of the tensor product Hilbert space $\Hi\otimes L^2(\B)$. As the representations $u$ and $\beta_1'$ of $\mathbb G_1$ on $\Hi$ resp. $L^2(\B)$ are continuous and $\Hi\stackrel[\Ga]{}{\boxtimes}L^2(\B)$ is the kernel of $u\otimes \id_{L^2(\B)}-\id_\Hi\otimes \beta_1'$, $\Hi\stackrel[\Ga]{}{\boxtimes}L^2(\B)$ is complete. 
\end{enumerate}

\end{proof}

\begin{proposition}\label{propdefhilbertspace}
We have 
\begin{enumerate}
\item $\Hi_x \stackrel[\Ga]{}{\boxtimes}L^2(\B)$ is isomorphic with $\Hi_{\varphi(x)} $ for all $x\in \Irred(\mathbb G_1)$.
\item \[\Hi\bokstimes[\Ga] L^2(\B)=\bigoplus_{\lambda\in \sigma(D)}V_{\lambda}\bokstimes[\Ga] L^2(\B)\] where $V_\lambda$ is the eigenspace of $\lambda \in \sigma(D)$. 
\item $V_\lambda\bokstimes[\Ga] L^2(\B)$ is finite dimensional for each $\lambda \in \sigma(D)$.
\end{enumerate}
\end{proposition}
Motivated by the first fact, we will call $\Hi_{\varphi(x)} $ the deformation of $\Hi_x$ for $x\in \Irred(\mathbb G_1)$.

\begin{proof}
\begin{enumerate}
\item Analogously to $\Hi\bokstimes[\Ga]L^2(\B)$, we define for $x\in \Irred(\mathbb G_1)$ $$\Hi_{x} \stackrel[\Ga]{}{\boxtimes}L^2(\B) =\{z\in \Hi_{x} \otimes L^2(\B)|U_{12}^{x}z_{13}=(\id\otimes \beta_1')(z)\}.$$
Now note that, for $x\in \Irred(\mathbb{G}_1)$ and $\xi \in \Hi_{\varphi(x)}$, $X^x(\xi\otimes\Lambda(1_\B)) \in \Hi_x \stackrel[\Ga]{}{\boxtimes}L^2(\B)$ and for $z \in \Hi_x \stackrel[\Ga]{}{\boxtimes}L^2(\B)$, $(\id_{\Hi_{\varphi(x)} }\otimes \omega_1')({X^x}^*z)\in \Hi_{\varphi(x)}$ where $\omega_1': L^2(\B)\to \mathbb{C}:\eta\mapsto \langle \Lambda(1),\eta\rangle$. Hence we can define the following maps: 
$$f_x:\Hi_{\varphi(x)} \to \Hi_x \stackrel[\Ga]{}{\boxtimes}L^2(\B):\xi\mapsto X^x(\xi\otimes\Lambda(1_\B))$$ 
$$g_x:\Hi_x \stackrel[\Ga]{}{\boxtimes}L^2(\B)\to \Hi_{\varphi(x)} :z\mapsto (\id_{\Hi_{\varphi(x)} }\otimes \omega_1')({X^x}^*z).$$ Using that $\beta_1'$ is ergodic (Proposition \ref{defhilbertspace}(2)), one can check that $g_x(z)\otimes \Lambda(1_{\B})={X^x}^*z$ which ensures that $f_x$ and $g_x$ are inverse to each other. Finally, using that $X^x$ is unitary, it is easy to see that $f_x$ and $g_x$ are also unitary.

\item Note first that as $D$ has compact resolvent, there exist a sequence  of real eigenvalues $(\lambda_n)_n$ with finite dimensional eigenspaces and such that $\lim_{n\to \infty}\lambda_n=\infty$. Hence we have $\Hi=\bigoplus_{\lambda\in\sigma(D)} V_\lambda$ and also $\Hi\otimes L^2(\B)=\bigoplus_{\lambda\in\sigma(D)} V_\lambda\otimes L^2(\B).$
As $U$ and $D$ commute, there is a subrepresentation $U_\lambda$ of $U$ on $V_\lambda$ for every eigenvalue $\lambda$ such that for $V_\lambda\bokstimes[\Ga]L^2(\B):=\{\xi\in V_\lambda\otimes L^2(\B)|(U_\lambda)_{12}\xi_{13}=(\id\otimes \beta_1')\xi\}$ we have $$\Hi\bokstimes[\Ga]L^2(\B)=\bigoplus_{\lambda\in \sigma(D)}V_\lambda\bokstimes[\Ga]L^2(\B).$$

\item Finally, decomposing $U_\lambda$ into irreducible representations of $\mathbb G_1$, we have $V_\lambda=\Hi_{x_1}\oplus\ldots\oplus \Hi_{x_l}$ for some $l\in \mathbb N$, $x_i\in \Irred(\mathbb G_1)$. Hence 
\begin{eqnarray}
V_\lambda\bokstimes[\Ga]L^2(\B)&=&(\Hi_{x_1}\oplus\ldots\oplus \Hi_{x_l})\bokstimes[\Ga]L^2(\B)\nonumber\\
&=&(\Hi_{x_1}\bokstimes[\Ga]L^2(\B))\oplus\ldots\oplus (\Hi_{x_l}\bokstimes[\Ga]L^2(\B)) \nonumber \\
&=&\Hi_{\varphi(x _1)}\oplus\ldots\oplus \Hi_{\varphi(x_l)}
\end{eqnarray}
where we used the first statement of this proposition. This last direct sum of finite dimensional Hilbert spaces implies $V_\lambda\bokstimes[\Ga]L^2(\B)$ to be finite dimensional. 
\end{enumerate}
\end{proof}

\begin{proposition}\label{defdiracop}
$D\otimes \id_{L^2(\B)}$ restricts to an unbounded selfadjoint operator $\tilde D$ on $\tilde{\Hi}=\Hi\bokstimes[\mathbb G_1]L^2(\B)$ of compact resolvent. 
\end{proposition}
\begin{proof}
As $D$ has compact resolvent, its restriction $D_\lambda$ to $V_\lambda$ is multiplication with $\lambda$ for every $\lambda$ in the spectrum. Therefore $D_\lambda\otimes \id$ can be restricted to $V_\lambda \bokstimes[\Ga]L^{2}(\B)\subset V_\lambda\otimes L^2(\B)$. Taking the direct sum we get an unbounded operator $\tilde D$ on $\Hi \bokstimes[\Ga]L^{2}(\B)$ with domain $\{(\xi_n)_n\in \bigoplus_{\lambda_n\in\sigma(D)}V_{\lambda_n} \bokstimes[\Ga]L^{2}(\B)|\sum_n|\lambda_n|^2\|\xi_n\|^2<\infty\}.$ By construction, we have $\tilde D= \sum_{\lambda\in \sigma (D)}\lambda (P_\lambda\otimes \id)$, where $P_\lambda$ is the projection $\Hi\to V_\lambda$. Hence it is of compact resolvent by proposition \ref{propdefhilbertspace}(3) and selfadjoint as $D$ is selfadjoint. Moreover, as $\Hi\otimes L^2(\B)=\bigoplus_{\lambda\in \sigma(D)} V_\lambda \otimes L^{2}(\B)$ it is the restriction of $$D\otimes \id_{L^2(\B)}=\bigoplus_{\lambda\in \sigma(D)} D_\lambda \otimes \id_{L^{2}(\B)} : \bigoplus_{\lambda\in \sigma(D)} V_\lambda \otimes L^{2}(\B) \to \bigoplus_{\lambda\in \sigma(D)} V_\lambda \otimes L^{2}(\B)$$ to $\Hi\bokstimes[\Ga]L^2(\B)$ concluding the proof.
\end{proof}

\begin{proposition}\label{defalgebra}
Define $\tilde{\A}=\A\stackrel[\OGa]{}{\boxdot}\B:= \{z\in\A\odot\B\ | (\alpha\otimes \id_\B)(z)=(\id_\A\otimes \beta_1)(z)\}$. Then $\tilde{\A}$ is a $^*$-algebra endowed with a coaction $\alpha_2=(\id\otimes \beta_2)_{|_{\tilde \A}}:\tilde\A\to \tilde\A\odot \OGb$ of $\OGb$. Moreover, $\tilde \A$ acts by bounded operators on $\tilde \Hi$: for $z\in \tilde{\A}$, we have $\tilde{L_z}: \tilde{\Hi}\to \tilde{\Hi}: v\mapsto zv$ by multiplication on $\B$ and action of $\A$ on $\Hi$ as a bounded operator on $\tilde{\Hi}$.

%Moreover $\tilde{\A}$ is SOT-dense in $\tilde{M}=M \stackrel[\LGa]{}{\overline{\boxtimes}} B'':= \{z\in M \vtimes B''\ | (\alpha\otimes \id_\B)(z)=(\id_\A\otimes \beta)(z)\}$.
\end{proposition}

\begin{proof}
The first statement is an application of theorem \ref{actiondeformation}. 
For the second, note that $\tilde \A\subset \A\odot \B$ and $\A\odot \B$ acts by bounded operators on $\Hi\otimes L^2(\B)$. Hence it suffices to prove that $\tilde{\A}$ leaves $\tilde{\Hi}$ invariant. Indeed, we have for $a\in \tilde{\A},\xi \in \tilde{\Hi}$
$$(\id_\Hi\otimes \beta'_1)(a\xi)=(\id_\A\otimes \beta_1)(a)(\id_\Hi\otimes \beta'_1)(\xi)=(\alpha\otimes \id_\B)(a)U_{12}\xi_{13}=U_{12}a_{13}U^*_{12}U_{12}\xi_{13}=U_{12}(a\xi)_{13}.$$

\end{proof}

\begin{theorem}
$(\tilde{\A},\tilde{\Hi},\tilde{D})$ constitues a spectral triple. 
\end{theorem}
\begin{proof}
Combining all the previous propositions, it suffices to prove that the commutator of $\tilde D$ with an element $a\in \tilde \A$ is bounded. For that, we will first prove that $\tilde \A$ leaves the domain of $\tilde D$ invariant and secondly we will proof that the commutator of $\tilde D$ and an arbitrary $a \in\tilde \A$ is bounded. 
%We will proceed in different steps. 
Let $z$ be an arbitrary element in $\A\odot \B$ and let $\xi$ be an arbitrary vector in $\dom(D\otimes \id)$. We will prove $z\xi \in\dom(D\otimes \id)$. As $\xi\in \dom(D\otimes \id)$, there exists a sequence $\xi_n$ in $\dom(D)\odot L^2(\B)$ such that simultaneously $\xi_n\to \xi$ and $(D\otimes \id)\xi_n\to (D\otimes \id)\xi$ for $n\to \infty$. Note that as $\A$ leaves the domain of $D$ invariant, $\A \odot \B$ leaves the core $\dom(D)\odot L^2(\B)$ of $D\otimes \id$ invariant and hence $z\xi_n\in \dom(D)\odot L^2(\B)$ for all $n$. Moreover, as $\A$ has bounded commutator with $D$, one can prove that $[D\otimes \id,z]$ is bounded on $\dom(D)\odot L^2(\B)$ and $(D\otimes \id)z(\xi_n)_n$ is a Cauchy sequence and thus converging. As $z\xi_n$ is an element of the core converging to $z\xi$ and $((D\otimes \id)z(\xi_n))_n$ converges, we know that $z\xi\in \dom(D\otimes \id)$ and $(D\otimes \id)z\xi_n\to (D\otimes \id)z\xi$. We can conclude that 
$(\A\odot \B)(\dom (D\otimes \id))\subset \dom(D\otimes \id)$ and it follows directly that $\tilde \A(\dom( \tilde D))\subset \dom(\tilde D)$.

Finally, we prove that $\tilde Dz-z \tilde D$ is indeed bounded on the domain of $\tilde D$. Let $\xi \in \dom(\tilde D)$ arbitrary and take a sequence $\xi_n\to \xi$ in $\dom(D)\odot L^2(\B)$. Then we know from above, that simultaneously
\begin{eqnarray*}
%\xi_n&\to& \xi,\\
%z\xi_n&\to &z\xi,\\
(D\otimes \id)z\xi_n&\to&(D\otimes \id)z \xi,\\
z(D\otimes \id)\xi_n&\to& z(D\otimes \id)\xi
\end{eqnarray*} and that $[D\otimes \id,z]$ is bounded on $\dom(D)\odot L^2(\B)$. Combining that, one can prove that indeed $\tilde D z - z \tilde D$ is  bounded on the domain.

\end{proof}

\begin{theorem} There exists a unitary representation $\tilde{U}$ of $\Gb$ on $\Hi\stackrel[\Ga]{}{\boxtimes}L^2(\B)$ such that $\mathbb G_2$ acts algebraically and by orientation-preserving isometries on $(\tilde{\A},\tilde{\Hi},\tilde{D})$ with $\tilde U$.

%with respect to the action $\alpha_2=(\id\otimes \beta_2)_{|_{\tilde{\A}}}:\tilde{\A}\to \tilde{\A}\otimes \OGb$.

\end{theorem}

\begin{proof}
Using the coaction $\beta_2:\B\to \B\odot \OGb$ and the CQG-action $\beta_2: B_u\to B_u\otimes C(\mathbb G_2)$, one can construct, along the lines of Lemma 5 in \cite{Boca1995} and the discussion above it, a representation $\tilde U_0\in \mathcal{M}(\mathcal{K}(L^2(\B))\otimes \Gb)$ such that 
$$\tilde U_0(\Lambda(b)\otimes a)=(\Lambda\otimes \id_{\Gb})\big(\beta_2(b)(1_\B\otimes a)\big).$$ Moreover, we know this is a unitary representation and furthermore,
\begin{equation}\label{beta2implementation}
\beta_2(b)=\tilde U_0(b\otimes \id)\tilde U_0^*.
\end{equation}
Now one can prove that $\id_\Hi\otimes \tilde U_0\in \mathcal{M}(\mathcal{K}(\Hi\otimes L^2(\B))\otimes \Gb)$ restricts to a representation $\tilde U \in \mathcal{M}(\mathcal{K}(\Hi\stackrel[\Ga]{}{\boxtimes}L^2(\B))\otimes \Gb)$. Indeed, as $\beta_1$ and $\beta_2$ commute, one has $(\beta_1'\otimes \id_{\Gb})\tilde U_0=(\id_{\Ga}\otimes \tilde U_0)(\beta_1'\otimes \id_{\Gb})$ and hence for $\xi\in \Hi\stackrel[\Ga]{}{\boxtimes}L^2(\B)$ and $a\in \Gb$ one has 
\begin{eqnarray*}
(\id_\Hi\otimes\beta_1'\otimes \id_{\Gb})(\id_\Hi\otimes \tilde U_0)(\xi\otimes a)&=&(\id_\Hi\otimes\id_{\Ga}\otimes \tilde U_0)(\id_{\Hi}\otimes\beta_1'\otimes \id_{\Gb})(\xi\otimes a)\\
&=&(\id_\Hi\otimes\id_{\Ga}\otimes \tilde U_0)(U\otimes\id_{L^2(\B)}\otimes \id_{\Gb})(\xi_{13}\otimes a)\\
&=&(U\otimes\id_{L^2(\B)}\otimes \id_{\Gb})\big((\id_\Hi\otimes \tilde U_0)(\xi\otimes a)\big)_{134}
\end{eqnarray*}

Then it suffices to prove that $\tilde{U}$ commutes with the Dirac operator of the deformed spectral triple and that there is a coaction of $\OGb$ on $\tilde{\A}$. 
As $\tilde D$ is the restriction of $D\otimes \id_{L^2(B)}$ and $\tilde U$ is the restriction of $ \id_\Hi\otimes \tilde U_0$, it follows directly that they commute. Using theorem \ref{actiondeformation}, we know that, given the coaction 
$$\alpha_1=\ad_{U}: \A\to \A\odot \OGa: a\to U(a\otimes \id_{\A})U^*,$$ 
there is a coaction $\alpha_2: \tilde{\A}\to\tilde{\A}\odot \OGb: z\to (\id_\A\otimes \beta_2)(z)$. Using \eqref{beta2implementation}, $\alpha_2=\id_\A\otimes \ad_{\tilde U_0}$ and regarding elements of $\A$ as operators on $\Hi$, we have $\alpha_2=\ad_{\tilde U}.$
 %It follows immediately that $(\id_\A\otimes \phi)\alpha_2(z)\in \tilde\A\subset \tilde\A''$ for $z\in \tilde{\A}$ and $\phi$ a state on $\Gb$.
\end{proof}

\begin{main}\label{mainthm}
Let $(\A,\Hi,D)$ be a compact spectral triple and let $\mathbb G_1=(\Ga,\Delta_1)$ be a compact quantum group acting algebraically and by orientation-preserving isometries on $(\A,\Hi,D)$ with a unitary representation $U$. Moreover let $\psi$ be a unitary fiber functor on $\mathbb G_1$. 
 \\ \\Then there exist a spectral triple $(\tilde\A,\tilde\Hi,\tilde D)$, a compact quantum group $\mathbb G_2= (\Gb,\Delta_2)$ monoidally  equivalent with $\mathbb G_1$ and a unitary representation $\tilde U$ of $\mathbb G_2$ on $\tilde\Hi$ such that the monoidal equivalence is associated to $\psi$ and $\mathbb G_2$ acts algebraically and by orientation-preserving isometries on the new spectral triple with $\tilde U$. 
 \\\\
 Denoting $\B$ to be the $\mathbb G_1 - \mathbb G_2$-bi-Galois object, one has
 \begin{equation}\label{mainequation}
 \tilde \A=\A\stackrel[\OGa]{}{\boxdot}\B,  \qquad \tilde\Hi=\Hi\stackrel[\Ga]{}{\boxtimes}L^2(\B), \qquad \tilde D = (D\otimes \id_{L^{2}(\B)})_{|_{\tilde \Hi}}.
 \end{equation}
\end{main}
In what follows, we will call this deformation procedure `monoidal deformation'.\\\\

To end this section, we will show that via the inverse monoidal equivalence on the deformed quantum group and spectral triple, one can obtain the original data again. 

\begin{theorem}
Let $(\A,\Hi,D)$ be a spectral triple, $\mathbb G_1$ a compact quantum group acting algebraically and by orientation-preserving isometries on $(\A,\Hi,D)$. Let $\psi$ be a unitary fiber functor, inducing a monoidal equivalence $\varphi:\mathbb G_1\to \mathbb G_2$ with bi-Galois object $\B$. Denote by $\varphi^{-1}:\mathbb G_2\to \mathbb G_1$ the inverse monoidal equivalence with bi-Galois object $\tilde \B$. Then
$$\Big(\A\boksdot[\OGa]\B \boksdot[\OGb] \tilde \B, \Hi \bokstimes[\Ga] L^2(\B)\bokstimes[\Gb] L^2(\tilde\B), D\otimes \id_{L^2(\B)}\otimes \id_{L^2(\tilde \B)}\Big)$$ is isomorphic with $(\A,\Hi,D)$ as spectral triples (definition \ref{isospectrtrpl}).
\end{theorem}

\begin{proof}
From proposition \ref{propGbtildeBB}, one obtains the following $^*$-isomorphisms:
$$\A\stackrel{\alpha_U}{\to} \A\boksdot[\OGa] \OGa\stackrel{\id\otimes \pi}{\to} \A\boksdot[\OGa] \B\boksdot[\OGb]\tilde \B$$ which are all compatible with the coaction of $\Gb$. Furthermore, recall the unitaries
$$f^{\varphi}_x:\Hi_{\varphi(x)} \to \Hi_x \stackrel[\Ga]{}{\boxtimes}L^2(\B):\xi^{\varphi(x)}\mapsto X^x(\xi^{\varphi(x)}\otimes\Lambda(1_\B))$$ for $x\in \Irred(\mathbb G_1)$ of proposition \ref{propdefhilbertspace}. Note that these unitaries intertwine the representations of $\mathbb G_2$ on the two Hilbert spaces. We then also have $$f^{\varphi^{-1}}_{\varphi(x)}:\Hi_{x} \to \Hi_{\varphi(x)} \stackrel[\Gb]{}{\boxtimes}L^2(\tilde \B):\eta^x\mapsto Z^{\varphi(x)}(\eta^x\otimes\tilde \Lambda(1_{\tilde\B}))$$ and combining them, we have a unitary:
$$ \theta^x: \Hi_{x} \to \Hi_x \stackrel[\Ga]{}{\boxtimes}L^2(\B) \stackrel[\Gb]{}{\boxtimes}L^2(\tilde \B): \eta^x\mapsto X^x_{12}Z_{13}^{\varphi(x)}(\eta^x\otimes\Lambda(1_\B)\otimes \tilde \Lambda(1_{\tilde\B})).$$ Denoting by $X$ and $Z$ resp. $\oplus_{x\in \Irred(\mathbb G_1)}X^x$ and $\oplus_{x\in  \Irred(\mathbb G_1)}Z^{\varphi(x)}$ (where we take the direct sum over the decomposition $\Hi=\oplus_{x\in \Irred(\mathbb G_1)}\Hi_x$), we then have a unitary $$\theta=\oplus_{x\in  \Irred(\mathbb G_1)} \theta_x:\Hi\to \Hi \bokstimes[\Ga] L^2(\B)\bokstimes[\Gb] L^2(\tilde\B):\xi\to X_{12}Z_{13}(\xi\otimes \Lambda(1_\B)\otimes \tilde \Lambda(1_{\tilde\B}))$$ and hence
$$(\id\otimes\pi)(\alpha_U(a))\theta(\xi)=(\id\otimes\pi)(U(a\otimes 1_{\Ga})U^*)X_{12}Z_{13}(\xi\otimes \Lambda(1_{\B}) \otimes \tilde\Lambda(1_{\tilde\B}))$$
\begin{equation}
=X_{12}Z_{13}(a\otimes 1_{\B}\otimes 1_{\tilde\B})Z_{13}^*X_{12}^*X_{12}Z_{13}(\xi\otimes \Lambda(1_{\B}) \otimes \tilde\Lambda(1_{\tilde\B}))=X_{12}Z_{13}(a\xi\otimes \Lambda(1_{\B}) \otimes \tilde\Lambda(1_{\tilde\B}))=\theta(a\xi)
\end{equation}
proving that $\theta(a\xi)=(\id\otimes\pi)(\alpha_U(a))(\theta(\xi))$. This concludes the proof.
\end{proof}

%\chapter{Main theorems}
%In this second chapter, we formulate two important theorems concerning our new method of deformation. First we prove that it is a generalization of 2-cocycle deformations of spectral triples, introduced by Goswami and Joardar in \cite{Goswami2014}. In the last section we prove that there is more: we construct an example which is a deformation along the lines we just proposed, but is not a 2-cocycle deformation.
%
\section{Cocycle deformation of spectral triples}
\label{cha-cocycle}

In this section we will fix a spectral triple $(\A,\Hi,D)$, a quantum group $\mathbb{G}$ acting algebraically on it by orientation-preserving isometries and a unitary fiber functor $\psi$ on $\mathbb{G}$ which satisfies $\dim(\Hi_x)=\dim(\Hi_{\psi(x)})$ for every $x\in \IrredG$. Unitary fiber functors which satisfy this condition will be called to be dimension-preserving and monoidal deformation via a dimension-preserving unitary fiber functor, a dimension-preserving monoidal deformation. Bichon et al. proved in \cite{Bichon2005} that dimension-preserving unitary fiber functors are in one-to-one correspondence with 2-cocycles on the dual quantum group. Using this, we will prove that dimension-preserving monoidal deformation is equivalent to the cocycle deformation introduced in \cite{Goswami2014}. 
%To start this chapter, we will first investigate some of the theory of cocycles and how to define a monoidal equivalence using a cocycle. Then we will prove that each monoidal equivalence which preserves the dimensions of the Hilbert spaces of the irreducible representations, comes indeed from such a cocycle.\\ \\
In this section we will frequently use slight adaptations of the work of Bichon et al. \cite{Bichon2005}.

\subsection{Cocycles on the dual of a compact quantum group}
Let $\mathbb{G}$ be a compact quantum group.
\begin{definition}\footnote{In \cite{Bichon2005}, the authors use another convention for cocycle. In fact, if $\Omega$ is a cocycle in our sense, $\Omega^*$ is one in the sense of Bichon and coauthors.}
Let $\mathbb{G}$ be a compact quantum group and $(c_0(\hat{\mathbb{G}}),\hat \Delta)$ its dual. We say a unitary element $\Omega \in \mathcal{M}(c_0(\Gdual)\otimes c_0(\Gdual))$ is a 2-cocycle on $\Gdual$ if it satisfies
\begin{equation}\label{cocycleprop}
(\Omega\otimes 1)(\hat \Delta\otimes\id)(\Omega)=(1\otimes \Omega)(\id\otimes \hat \Delta)(\Omega).
\end{equation}
\end{definition}
Denoting for $x\in \IrredG, p_x$ to be the projection $c_0(\Gdual)\to B(\Hi_x)$, we will say a cocycle is normalized if $(p_\varepsilon\otimes \id)\Omega=p_\varepsilon\otimes\id$ and $(\id\otimes p_\varepsilon)\Omega=\id\otimes p_\varepsilon$. From now on we will always assume 2-cocycles to be normalized.\\

\begin{proposition}[\cite{Bichon2005}]\label{cocycle-unfibfunct}
Let $\Omega$ be a normalized unitary 2-cocycle on $\Gdual$ and denote $$\Omega_{(2)}=(\Omega\otimes 1)(\hat \Delta\otimes\id)(\Omega)=(1\otimes \Omega)(\id\otimes \hat \Delta)(\Omega).$$
Then there exists a unique unitary fiber functor $\psi_\Omega$ on $\mathbb G$ such that $$\Hi_{\psi_\Omega(x)}=\Hi_x,\qquad \psi_\Omega(S)=\Omega S,\qquad \psi_\Omega(T)=\Omega_{(2)}T$$ for all $S\in \Mor(y\otimes z,x)$ and $T\in \Mor(x\otimes y\otimes z,a)$ and $x,y,z\in \IrredG$. Moreover it is dimension-preserving.
\end{proposition}
\begin{proof}
The proof follows directly as our $\psi$ satisfies the conditions of remark \ref{unitfibfunctinducesmoneq}. That it is dimension-preserving, follows directly by construction.
\end{proof}

Using this unitary fiber functor, one can make a new compact quantum group $\mathbb G_\Omega=(C(\mathbb G_\Omega),\Delta_\Omega)$ \cite{Bichon2005} and a monoidal equivalence $\varphi: \mathbb G \to \mathbb G_\Omega$ along the lines of proposition \ref{unfibfunctismoneq}. Note that the dual quantum group will be $(c_0(\Gdual_\Omega),\hat \Delta_\Omega)$ where $$c_0(\Gdual_\Omega)=\bigoplus_{x\in \IrredG} B(\Hi_x)=c_0(\Gdual)$$ and $$\hat \Delta_\Omega(a)\psi_\Omega(S)=\psi_\Omega(S)a$$ where $\hat \Delta_\Omega(a)=\Omega \Delta(a)\Omega^*$.

\begin{proposition}[\cite{Bichon2005}]\label{dimpresunfibfunctinducescocycle}
For every dimension-preserving unitary fiber functor $\psi$ on a quantum group $\mathbb G$, there exists a normalized unitary 2-cocycle $\Omega$ on $\Gdual$ such that $\psi\cong\psi_\Omega$.
\end{proposition}
\begin{proof}
The proof is a slightly adapted version of the proof of proposition 4.5 in \cite{Bichon2005}.
\end{proof}
This theorems tells us that every dimension-preserving monoidal equivalence comes from a cocycle. The next step to prove that a dimension-preserving monoidal deformation of a spectral triple is a cocycle deformation is to introduce the algebraic notion of a 2-cocycle. We will proof that every 2-cocycle on the dual of a compact quantum group induces an algebraic 2-cocycle on the compact quantum group and that the monoidal deformation is equivalent to a cocycle deformation of the spectral triple as was introduced by Goswami in \cite{Goswami2014}.

\subsection{Algebraic 2-cocycle deformation of a spectral triple}
%\subsection{Algebraic 2-cocycles}
We will start with defining the algebraic counterpart of a 2-cocycle on the dual of a compact quantum group. In algebraic literature (for example Schauenburg \cite{Schauenburg2004}), the definition and theorems are stated for Hopf algebras. We make slight adaptations to Hopf $^*$-algebras.
\begin{definition}\label{defdualcocycle}
Let $H$ be a Hopf algebra. 
\begin{enumerate}
\item An (algebraic) dual 2-cocycle on $H$ is a linear map $\sigma: H\odot H \to \mathbb{C}$ such that $$\sigma(a_{(1)},b_{(1)})\sigma(a_{(2)}b_{(2)},c)=\sigma(b_{(1)},c_{(1)})\sigma(a,b_{(2)}c_{(2)})$$ for all $a,b,c\in H$. Moreover, a dual 2-cocycle is called normalized  if $\sigma(1,h)=\sigma(h,1)=\varepsilon(h)$ for all $h\in H$.
\item A dual 2-cocycle is called invertible if there exists a linear map $\sigma':H\odot H\to \mathbb C$ such that $$\sigma(a_{(1)},b_{(1)})\sigma'(a_{(2)},b_{(2)})=\varepsilon(a)\varepsilon(b)=\sigma'(a_{(1)},b_{(1)})\sigma(a_{(2)},b_{(2)}).$$ In this case, $\sigma'$ is called the inverse dual cocycle and written $\sigma^{-1}$. Moreover $\sigma^{-1}$ satisfies
$$\sigma^{-1}(a_{(1)}b_{(1)},c)\sigma^{-1}(a_{(2)},b_{(2)})=\sigma^{-1}(a,b_{(1)}c_{(1)})\sigma^{-1}(b_{(2)},c_{(2)}).$$
\item If $H$ is a Hopf $^*$-algebra, a dual 2-cocycle $\sigma$ is called unitary if it satisfies
$$\overline{\sigma (a,b)}=\sigma^{-1}(S(a)^*,S(b)^*).$$ In that case, we also have
$$\overline{\sigma^{-1} (a,b)}=\sigma(S(a)^*,S(b)^*).$$
%\item If $H$ is a Hopf $^*$-algebra, a dual 2-cocycle $\sigma$ is called real if it satisfies
%$$\overline{\sigma (a,b)}=\sigma(S^2(b)^*,S^2(a)^*).$$ In that case, we also have
%$$\overline{\sigma^{-1} (a,b)}=\sigma^{-1}(S^2(b)^*,S^2(a)^*).$$
\end{enumerate}
\end{definition}
In the rest of the section, when we use 2-cocycles on Hopf $^*$-algebras, we will always assume them to be unitary.

Using such a dual 2-cocycle, we can make a new $^*$-algebra and several new $H$-comodule $^*$-algebras. We will use the following linear maps:
\begin{itemize}
\item $U: H\to \mathbb C: h\mapsto \sigma(h_{(1)},S(h_{(2)}))$,
\item $V: H\to \mathbb C: h\mapsto U(S^{-1}(h))$.
%\item $W:H\to \mathbb C: h \mapsto U(S^{-2}(h))$,
%\item $V: H\to \mathbb C: h \mapsto W(h_{(1)}) W^{-1}(S(h_{(2)}))$.
\end{itemize}
One can prove that for $U^{-1}(h)=\sigma^{-1}(S(h_{(1)}),h_{(2)})$ and $V^{-1}(h)=U^{-1}(S^{-1}(h))$ one has $U(h_{(1)})U^{-1}(h_{(2)})=\varepsilon(h)=U^{-1}(h_{(1)})U(h_{(2)})$ and $V(h_{(1)})V^{-1}(h_{(2)})=\varepsilon(h)=V^{-1}(h_{(1)})V(h_{(2)})$. 

\begin{definition}\label{defHopfdef}
Given an invertible dual 2-cocycle $\sigma$ on a Hopf $^*$-algebra $(H,\Delta,\varepsilon,S, ^*)$, we define $(H^\sigma,\Delta_\sigma,\varepsilon_\sigma,S_\sigma, ^{*_\sigma})$ to be the twisted Hopf $^*$-algebra which 
\begin{itemize}
\item is isomorphic to $H$ as a co-algebra,
\item has multiplication defined by $g\cdot_\sigma h =\sigma(g_{(1)},h_{(1)})g_{(2)}h_{(2)}\sigma^{-1}(g_{(3)},h_{(3)})$,
\item has antipode $S_\sigma(h)=U(h_{(1)})S(h_{(2)})U^{-1}(h_{(3)})$,
\item has counit $\varepsilon_\sigma=\varepsilon$
\item and has involution %\footnote{We claim that the formulas (2.26) in \cite{Majid} are not correct. This is easy to see: the maps $U$ and $\theta$ there are linear, which gives a problem in the definition of the deformed involution. The mistake is being made in dualizing proposition 2.3.7, as the proof of that proposition is correctly given. One can check that the new defined involution $h^{*_{\sigma}}=\overline{V^{-1}(h_{(1)})}h_{(2)}^*\overline{V(h_{(3)})}$ is a well defined involution on $H^{\sigma}$ and that it is the dual of the deformed involution $^{*_{\chi}}=(S^{-1}U)((\cdot))^*(S^{-1}U^{-1})$, defined in proposition 2.3.7 of \cite{Majid}.} 
$h^{*_\sigma}=V^{-1}(h_{(1)}^*)h^*_{(2)}V(h_{(3)}^*)$.
\end{itemize}
\end{definition}

\begin{definition}\label{defcomoddef}
We define 
\begin{enumerate}
\item $\mathbb C \#_\sigma H$ to be a $H^\sigma-H$-bicomodule $^*$-algebra which
\begin{itemize}
\item is isomorphic to $H$ as right $H$-comodule,
\item has twisted multiplication $(1\# g)(1\# h)=\sigma(g_{(1)},h_{(1)})\# g_{(2)}h_{(2)}$,
\item has a coaction $\beta_1:\mathbb C \#_\sigma H\to H^\sigma \odot (\mathbb C \#_\sigma H):(1 \# h)\mapsto h_{(1)}\otimes (1 \# h_{(2)})$,
\item and has involution $(1\#h)^{*_{\mathbb C \#_\sigma H}}=1\#V^{-1}(h^*_{(1)})h^*_{(2)}$.
\end{itemize}

 and \item $H \tensor[_{\sigma^{-1}}]{\#}{} \mathbb C$ to be a $H-H^\sigma$-bicomodule algebra which
 \begin{itemize}
\item is isomorphic to $H$ as left $H$-comodule,
\item has twisted multiplication $(g\# 1)(h\# 1)=g_{(1)}h_{(1)}\# \sigma^{-1}(g_{(2)},h_{(2)})$,
\item has a coaction $\beta_2: H \tensor[_{\sigma^{-1}}]{\#}{} \mathbb C\to  (H \tensor[_{\sigma^{-1}}]{\#}{} \mathbb C) \odot H^\sigma:(h \# 1)\mapsto (h_{(1)} \# 1) \otimes h_{(2)}$,
\item and has involution $(h\#1)^{*_{H \tensor[_{\sigma^{-1}}]{\#}{} \mathbb C}}=h^*_{(1)}V(h^*_{(2)})\#1$.
\end{itemize} 
%\item[3] If $A$ is a right (resp. left) $H$-comodule-algebra, $A$
\end{enumerate}
\end{definition}

%\begin{proposition}
%Let $H$ be a Hopf $^*$-algebra and $\sigma$ an invertible dual 2-cocycle on $H$. Then $H \tensor[_{\sigma^{-1}}]{\#}{} \mathbb C$ is a $H-H^\sigma$-bi-Galois-object with inverse $\mathbb C \#_\sigma H$.
%\end{proposition}

%\begin{proof}
%
%\end{proof}

\begin{definition}
Let $H$ be a Hopf $^*$-algebra and $\sigma$ an invertible dual 2-cocycle on $H$. Let $A$ be a right $H$-comodule $^*$-algebra with coaction $\alpha:A\to A\odot H$. We define $A \tensor[_{\sigma^{-1}}]{\#}{} \mathbb C$ to be a right $H^\sigma$-comodule $^*$-algebra which
\begin{itemize}
\item is isomorphic to $A$ as vector space,
\item has multiplication $(a\# 1)(a'\# 1)=a_{(0)}a'_{(0)}\#\sigma^{-1}(a_{(1)},b_{(1)})$,
\item has a coaction $\tilde \alpha:A \tensor[_{\sigma^{-1}}]{\#}{} \mathbb C \to (A \tensor[_{\sigma^{-1}}]{\#}{} \mathbb C) \odot H^\sigma :(a\#1)\mapsto (a_{(0)} \#1)\otimes a_{(1)}$, 
\item and has involution $(a\#1)^{*_{A \tensor[_{\sigma^{-1}}]{\#}{} \mathbb C}}=a_{(0)}V(a^*_{(1)})\#1$.
\end{itemize}
\end{definition}

\begin{theorem}\label{BHalg}
Let $H$ be a 
%$k$-flat 
Hopf $^*$-algebra
%. Then the following are equivalent:
%\begin{itemize}
%\item $B \cong H$ as left $H$-comodules
%%\item $B$ is cleft, i.e. there is a left  $H$-colinear and convolution invertible map $H\to B$
%\item $B\cong H_\sigma^{-1}\# \mathbb C$ for a convolution invertible 2-cocycle $\sigma : H \odot H\to k$.
%\end{itemize} If $H$ and $B$ satisfy these equivalent conditions, we have
%\begin{enumerate}
%
%\item  If $B$ is cleft (such that $B\cong  H _\sigma^{-1}\# k$), then $$H^\sigma\cong \tilde{H}$$ as defined in chapter \ref{cha-algdef} and $B \cong k\#_\sigma H^\sigma$ as a right $H_\sigma$-comodule algebra.
%\item If 
and $A$ a right $H$-comodule $^*$-algebra with coaction $\alpha:A\to A\odot H$. Denote $B= H \tensor[_{\sigma^{-1}}]{\#}{} \mathbb C$. Then $$A\boksdot B \cong A \tensor[_{\sigma^{-1}}]{\#}{} \mathbb C.$$ 
%\end{enumerate}
\end{theorem}

\begin{proof}
We have the natural $^*$-algebraic isomorphisms
$$A\stackrel{\alpha}{\to} A\boksdot[H] H\stackrel{\id\otimes \varepsilon}{\to}A.$$
Using it as vector space isomorphisms, deforming the multiplications and using that $B$ and $H$ are isomorphic as left $H$-comodules, it is easy to check that we have a well defined $^*$-algebra isomorphism $$ \lambda: A \tensor[_{\sigma^{-1}}]{\#}{} \mathbb C\to A\boksdot B: (a\#1)\to a_{(0)}\otimes (a_{(1)}\#1).$$
\end{proof}

%\subsection{Algebraic 2-cocycle deformation as defined by Goswami - Joardar}
In this paragraph we give a slightly adapted version of a result of Goswami and Joardar in \cite{Goswami2014}.
\begin{theorem}[\cite{Goswami2014}\footnote{We want to note that Goswami erroneously referred to \cite{Majid1995} to explain the deformation of the Hopf $^*$-algebra. Indeed, Majid uses a reality condition and Goswami a unitarity condition, which makes the theory of Majid not applicable here. We developed a new deformation of the star structure using a unitary cocycle which results in definitions \ref{defHopfdef} and \ref{defcomoddef}.}]
Let $(\A,\Hi,D)$ be a spectral triple and $\mathbb G$ a compact quantum group acting on it algebraically and by orientation-preserving isometries with the representation $U$. Let $\sigma$ be an (algebraic) unitary dual 2-cocycle on $\OG$. Then
\begin{enumerate}[(a)]
\item there exists a representation $\pi_\sigma:\A \tensor[_{\sigma^{-1}}]{\#}{} \mathbb C\to B(\Hi)$
\item $(\A \tensor[_{\sigma^{-1}}]{\#}{} \mathbb C,\Hi,D)$ is a spectral triple.
\end{enumerate}
\end{theorem}
\begin{proof}
\begin{enumerate}[(a)]
\item Denote the the coaction $\alpha=\ad_{U}$ of $\OG$ on $ \A \tensor[_{\sigma^{-1}}]{\#}{} \mathbb C$ by $\alpha(a)=a_{(0)}\otimes a_{(1)}$. Let $\mathcal N$ be a dense subspace of $\Hi$ such that $U(\mathcal N)\subset \mathcal N \odot \OG$ and on that subspace, let $U(\xi)=\xi_{(0)}\otimes \xi_{(1)}$. Then we can define, for $a\in \A \tensor[_{\sigma^{-1}}]{\#}{} \mathbb C $:
$$\pi_\sigma(a):   \Hi\to \Hi:  \xi\mapsto a_{(0)}\xi_{(0)}\sigma^{-1}(a_{(1)},\xi_{(1)}).$$ In section 4.3 of \cite{Goswami2014} in it is proved that $\pi_\sigma(a)$ is bounded for all $a\in \A \tensor[_{\sigma^{-1}}]{\#}{} \mathbb C $ and that $\pi_\sigma$  is a well defined $^*$-morphism.
\item This is theorem 4.10(4) in \cite{Goswami2014}.
\end{enumerate}

\end{proof}

%\begin{proposition}
%Let $H$ be a Hopf algebra with right Galois object $B$. Suppose furthermore that 
%\end{proposition}

\subsection{Linking dimension-preserving monoidal equivalences with algebraic cocycles}
In proposition \ref{dimpresunfibfunctinducescocycle}, we proved that there is an equivalence between dimension-preserving unitary fiber functors on  a compact quantum group $\mathbb G$ and cocycles on the dual $\hat{\mathbb G}$. In the following theorem \ref{Omegainducessigma}, we will prove that there is also an equivalence between cocycles on $\hat{\mathbb G}$ and (algebraic) dual cocycles on $\OG$. Moreover, we will show in theorem \ref{multkappa} that the bi-Galois object $\B$ associated with the monoidal equivalence induced by the fiber functor, will be of the form $\B=\OG\ _{\sigma^{-1}}\#\mathbb C$.\begin{theorem}\label{Omegainducessigma}
Let $\mathbb{G}$ be a compact quantum group. If $\Omega$ is a unitary 2-cocycle on the dual $\Gdual$, the formula
\begin{equation}\label{cocycleeq}
\sigma(u_{ij}^x\otimes u^y_{kl})=\langle \xi^x_{i}\otimes \xi_k^y,\Omega(\xi^x_j\otimes \xi_l^y)\rangle, x,y \in \IrredG
\end{equation} 
defines a unique (algebraic) unitary dual 2-cocycle $\sigma$ on $\OG$. On the other hand, if $\sigma$ is an (algebraic) unitary dual 2-cocycle on $\OG$, formula \ref{cocycleeq} uniquely defines a unitary 2-cocycle $\Omega$ on $\hat{\mathbb G}$. 
\end{theorem}

\begin{proof}
Under the first assumption, as the $u^x_{ij}$ constitute a basis of $\OG$, the linear map $\sigma$ is well defined. Using the cocycle property \eqref{cocycleprop} of $\Omega$, one can check that $\sigma$ satisfies the dual cocycle condition in definition \ref{defdualcocycle}(1). It is normalized and unitary as $\Omega$ is normalized and unitary. Under the second assumption, $\Omega$ is uniquely and well defined as element of $\mathcal M(c_0(\hat{\mathbb G}) \otimes c_0(\hat{\mathbb G}))$. The dual cocycle condition in definition \ref{defdualcocycle}(1) will imply the cocycle condition \eqref{cocycleprop} of $\Omega$. Again, $\Omega$ is normalized and unitary as $\sigma$ is.
\end{proof}

Remark that, as $\Omega^*$ is the inverse of $\Omega$, we see that $\sigma'$ associated with $\Omega^*$ is the convolution inverse of $\sigma$. We will denote it with $\sigma^{-1}$ and we have $$
\sigma^{-1}(u_{ij}^x\otimes u^y_{kl})=\langle \xi^x_{i}\otimes \xi_k^y,\Omega^*(\xi^x_j\otimes \xi_l^y)\rangle.$$

\begin{theorem}\label{multkappa}
Let $\mathbb{G}$ be a compact quantum group with a dimension-preserving unitary fiber functor $\psi$. Let $\B$ be the bi-Galois object associated to $\psi$ with coaction $\beta_1:\B\to \OGa\odot \B$, let $\Omega$ be the unitary 2-cocycle on the dual $\Gdual$ associated to $\psi\cong\psi_\Omega$ and $\sigma$ the algebraic dual 2-cocycle equivalent with $\Omega$ (proposition \ref{Omegainducessigma}). Then there exists a $^*$-algebra isomorphism
%the vector spaces $\OG$ and $\B$ are isomorphic. Moreover, if we denote $\chi: \OG\to \B$ the isomorphism, then $$\chi(u^x_{ij}u^y_{st})=\sum_{k,l}\chi(u^x_{ik})\chi(u^y_{sl})\sigma(u^x_{kj},u^y_{lt})$$ and
% $$\chi(u^x_{ij})\chi(u^y_{st})=\sum_{k,l}\chi(u^x_{ik}u^y_{sl})\sigma^{-1}(u^x_{kj},u^y_{lt})$$ for $x,y\in \IrredG$ which makes 
$$\chi:\B\to \OG\ _{\sigma^{-1}}{\#} \mathbb{C}$$ such that $(\id\otimes \chi)\beta_1=\Delta\circ \chi$.
\end{theorem}

\begin{proof}
Denoting $\varphi: \mathbb G\to \mathbb G_\Omega$ to be the monoidal equivalence associated to $\psi$, we can find unitaries $u_x=\Hi_{x}\to \Hi_{\varphi(x)}$, as $\dim(\varphi(x))=\dim(x)$ for all $x\in \IrredG$%\zeta_i\to \xi_i$
. Fixing a $x\in \IrredG$, we can define $Y^x=X^x(u_x\otimes 1)\in B(\Hi_x)\odot \B$ and $$Y'=\oplus_{x\in \Irred(\mathbb G)}Y^x\in \mathcal{M}(c_0(\Gdual)\otimes B_r)$$ (where we take the direct sum over all classes, all of them with multiplicity one). Note that the matrix coefficients of the $X^x$ constitute a basis of $\B$ by theorem \ref{monoidalGalois}. As the $u_x$ are unitaries, also the matrix coefficients of the $Y^x$ (let's call them $b_{ij}^x$) and hence of $Y'$ form a basis of $\B$. As both the $(u^x_{ij})_{ij,x}$ and $(b^x_{ij})_{ij,x}$ are bases of $\OG$ resp. $\B$, we have a vector space isomorphism
$$\chi: \OG\to \B: u^x_{ij}\mapsto b^x_{ij}$$ which is compatible with the coactions (i.e. $(\id\otimes \chi) \Delta=\beta_1\circ \chi$). Moreover, one can prove that, analogously as in the proof proposition 4.5 of \cite{Bichon2005}, $Y'$ satisfies the equation
\begin{equation}\label{YOmegarep}
(\hat \Delta\otimes\id)(Y')=Y'_{13}Y'_{23}(\Omega\otimes 1_\B).
\end{equation} 
As
$(\hat \Delta\otimes \id)(Y')=(\hat \Delta\otimes \chi)(\mathbb{V})$ by construction and $(\hat \Delta\otimes \id)(\mathbb{V})=\mathbb{V}_{13}\mathbb{V}_{23}$ by definition of $\mathbb V$, one can prove that
\begin{eqnarray*}
\chi(u^x_{ij}u^y_{st})&=&\langle \xi^{x}_i\otimes  \xi^y_{s}\otimes \id,(\hat \Delta\otimes \id)(Y')(\xi^x_j\otimes \xi^y_t\otimes \id)\rangle\\
&=&\langle \xi^{x}_i\otimes  \xi^y_{s}\otimes \id,(Y'_{13}Y'_{23}(\Omega\otimes \id))(\xi^x_j\otimes \xi^y_t\otimes \id)\rangle\\
&=&\sum_{p,q}\chi(u^x_{ip})\chi(u^y_{sq})\sigma(u^x_{pj},u^y_{qt})
\end{eqnarray*} where we used theorem \ref{Omegainducessigma}.
Hence, also $\chi(u^x_{ij})\chi(u^y_{st})=\sum_{k,l}\chi(u^x_{ik}u^y_{sl})\sigma^{-1}(u^x_{kj},u^y_{lt}),$ which means 
\begin{equation}\label{eqchi}
\chi(a)\chi(b)=\chi(a_{(0)}b_{(0)})\sigma^{-1}(a_{(1)},b_{(1)}).
\end{equation} 
Finally, to check that $\chi$ is a $^*$-algebra isomorphism, note that by the previous equation, we also have
$$\chi(ab^*)=\chi(a_{(0)})\chi(b_{(0)}^*)\sigma(a_{(1)},b_{(1)}^*)$$ and hence
$$\chi(u^x_{ij})^*=\sum_{k,l}\chi(u^x_{kj})^*\chi\big(u^x_{kl}(u^x_{il})^*\big)=\sum_{k,l,p,q}\chi(u^x_{kj})^*\chi(u^x_{kp})\chi\big((u^x_{iq})^*\big)\sigma\big(u^x_{pl},(u^x_{ql})^*\big)$$ $$=\sum_{l,q}\chi\big((u^x_{iq})^*\big)\sigma\big(u^x_{jl},(u^x_{ql})^*\big)$$
by unitarity of the $U^x$ and the $Y^x$, which implies 
\begin{equation}
\chi(a)^*=\chi(a_{(1)}^*)\sigma(S(a_{(3)})^*,a_{(2)}^*)=\chi(a_{(1)}^*)V(a_{(2)}^*)
\end{equation}
where $V(a)=\sigma(S^{-1}(a_{(2)}),a_{(1)})$ as before. This proves the last statement.
\end{proof}

\subsection{Dimension-preserving monoidal deformation is isomorphic to algebraic 2-cocycle deformation}
In this last paragraph of section \ref{cha-cocycle}, we state and prove the main result of this section: the Goswami-Joardar cocycle deformation amounts to our deformation with a dimension-preserving monoidal equivalence. 
\begin{theorem}\label{maincocycle}
Let $(\A,\Hi,D)$ be a spectral triple, $\mathbb G$ a compact quantum group acting on it algebraically and by orientation-preserving isometries via a unitary representation $U$ and let $\psi$ be a dimension-preserving unitary fiber functor on $\mathbb G$. Denoting by $\B$ the corresponding bi-Galois object, there exists an (algebraic) unitary dual 2-cocycle $\sigma$ such that $(\A\boksdot[\OG] \B,\Hi\bokstimes[\G] L^2(\B),\tilde D)$ defined in section \ref{sec-mondef} and $( \A \tensor[_{\sigma^{-1}}]{\#}{} \mathbb C ,\Hi,D)$ are isomorphic as spectral triples.
\end{theorem}

Remember that $\B$ is the bi-Galois object associated to the fiber functor $\psi$, $L^2(\B)$ the GNS-space with respect tot the invariant state $\omega=(h\otimes \id)\beta_1$ and the deformed Dirac operator $\tilde D$ from section \ref{sec-mondef}.
We give the proof via some propositions.

\begin{proposition}
\begin{enumerate}
\item %Making the decomposition $\Hi=\sum_{x\in \Irred{\mathbb{G}_2}}\Hi_{\varphi(x)}$, we have 
There exists a unitary $Y\in \mathcal M(\mathcal K(\Hi)\otimes B_r)$ such that $\phi:\Hi\to\Hi\bokstimes[\G] L^2(\B):\xi\to Y(\xi\otimes 1)$ is an isomorphism of Hilbert spaces.
\item Under this isomorphism, $\phi D=\tilde D  \phi$.
\item $\A\stackrel[\OG]{}{\boxdot} \B\cong \A\ _{\sigma^{-1}}\#\mathbb C$ with $\sigma$ the algebraic dual 2-cocycle associated to the dimension-preserving unitary fiber functor $\psi$.
\end{enumerate}
\end{proposition}

\begin{proof}
\begin{enumerate}
\item Recall the unitaries $u_x:\Hi_{x}\to\Hi_{\varphi(x)}$ from the proof of theorem \ref{multkappa} and the mutually inverse unitaries $$f_x:\Hi_{\varphi(x)}\to \Hi_{x} \stackrel[\G]{}{\boxtimes}L^2(\B):\xi\mapsto X^{x}(\xi\otimes\Lambda(1_\B))$$ and $$g_x:\Hi_{x} \stackrel[\G]{}{\boxtimes}L^2(\B)\to \Hi_{\varphi(x)}:z\mapsto (\id_{\Hi_{\varphi(x)}}\otimes \omega_1')({X^{x}}^*z)$$  from the proof of proposition \ref{propdefhilbertspace} point 1. Therefore, defining $\phi_x=f_x\circ u_x=:\Hi_{x}\to \Hi_{x} \stackrel[\G]{}{\boxtimes}L^2(\B)$, $\phi'_x=u_x^*\circ g_x:\Hi_{x} \stackrel[\G]{}{\boxtimes}L^2(\B)\to \Hi_{x}$ and $Y= \oplus_{x\in \IrredG} Y^x$ we can make $\phi=\sum_{x\in \Irred(\mathbb{G})}\phi_x$ (where in both cases we take the sum over the irreducible representations appearing in the decomposition of $U$) such that $\phi(\xi)=Y(\xi\otimes 1)$ for $\xi\in \Hi$. $Y$ is unitary and hence $\phi$ is the desired isomorphism of Hilbert spaces.
\item We have to prove that $\phi(D\xi)=\tilde D(\phi(\xi))$ for $\xi \in \dom (D)$. Denoting $P_{\lambda_n}$ resp. $\tilde P_{\lambda_n}$ to be the projection onto $V_{\lambda_n}$ resp. $V_{\lambda_n} \bokstimes[\G]L^2(\B)$, note that, as $Y=(\id\otimes \chi)(U)$ and $U$ commutes with $D$, $\phi(P_{\lambda_n}\xi)=\tilde P_{\lambda_n}(\phi(\xi))$. Then 
$$\sum_n|\lambda_n|^2\|\tilde P_{\lambda_n}(\phi(\xi))\|^2=\sum_n|\lambda_n|^2\|\phi(P_{\lambda_n}(\xi))\|^2=\sum_n|\lambda_n|^2\| P_{\lambda_n}(\xi)\|^2<\infty$$ as $\xi\in \dom(D)$ and hence $\phi$ maps the domain of $D$ into the domain of $\tilde D$. Also, by the previous remark, trivially, $\tilde D_n=\tilde D_{|_{V_{\lambda_n} \bokstimes[\G]L^2(\B)}}$ commutes with $\phi$ for all $n$. Taking the direct sum,
we can conclude that also $\tilde D$ commutes with $\phi$.

\item The proof follows from theorem \ref{BHalg} and theorem \ref{multkappa}.
\end{enumerate}
\end{proof}
Finally, it suffices to prove that the actions of the algebras on the Hilbert spaces are isomorphic.
\begin{proposition}
The action of $\A\ _{\sigma^{-1}}\#\mathbb C$ on $\Hi$ is isomorphic to the action of  $\A\stackrel[\OG]{}{\boxdot} \B$ on $\Hi\bokstimes[\G] L^2(\B)$ i.e. if $\phi:\Hi\to \Hi \stackrel[\G]{}{\boxtimes} L^2(\B)$ and $(\id\otimes \chi)\alpha: \A\ _{\sigma^{-1}}\#\mathbb C\to\A\stackrel[\OG]{}{\boxdot} \B$ are the isomorphisms of the previous proposition, we have:
$$\phi(a\cdot_\sigma \xi)=(\id\otimes \chi)\alpha(a)\phi(\xi).$$
\end{proposition}

\begin{proof}
Let $a\in \A$ and let $\xi^{z,m}_n$ be the $n$-th basisvector in the $m$-th summand of $\Hi_z$ in the decomposition of $\Hi$. Using the Hilbert space isomorphism $\phi : \Hi\to \Hi \bokstimes[\G] L^2(\B)$ and the notation $a\cdot_\sigma \xi^{z,m}_n$ for the deformed action of $a\#1\in \A\ _{\sigma^{-1}}\#\mathbb C$ on $\xi$, we will prove that $$\phi(a\cdot_\sigma \xi)=(\id\otimes \chi)\alpha(a)\phi(\xi)$$ by proving $$a\cdot_\sigma \xi=Y^*(\id\otimes \chi)\alpha(a)(Y(\xi\otimes 1)).$$ 
First we compute $a\cdot_\sigma \xi^{z,m}_n$.
Writing 
\begin{equation}\label{alphaa}
U(\xi_j^{x,k}\otimes \id)=\sum_{i}\xi^{x,k}_i\otimes u^{x,k}_{ij},
%$$\langle \xi^{x,k}_s|a|\xi^{y,l}_{t}=a^{x,k,s}_{y,l,t}$$ and 
\end{equation} and noting that $\alpha_U(a)=U(a\otimes 1)U^*$ and that $a\cdot_\sigma \xi= a_{(0)}\xi_{(0)}\sigma^{-1}(a_{(1)},\xi_{(1)})$ where $U(\xi\otimes 1)=\xi_{(0)}\otimes \xi_{(1)}$ it is only a calculation to check that  
\begin{equation}
a\cdot_\sigma \xi_n^{z,m} =\sum_{x,k,i,j,q}\xi^{x,k}_i\langle\xi_j^{x,k},a\xi^{z,m}_q\rangle\sum_{s}\sigma^{-1}( u^{x,k}_{ij}(u^{z,m}_{sq})^*,u^{z,m}_{sn}) 
\end{equation} which is a finite sum as $(\alpha_U)_{|_{\A}}$ is an algebraic coaction.
Using, moreover, the cocycle relations, we get
\begin{equation}\label{calcgosw}
a\cdot_\sigma \xi_n^{z,m}=\sum_{x,k,i,j,q}\xi^{x,k}_i\langle\xi_j^{x,k},a\xi^{z,m}_q\rangle\sum_{t,r}\sigma^{-1}((u^{z,m}_{tr})^*,u^{z,m}_{tn})\sigma(u^{x,k}_{ij},(u^{z,m}_{rq})^*).
\end{equation}
Next, we will compute $Y^*(\id\otimes \chi)\alpha_U(a)Y(\xi^{z,m}_n\otimes 1)$.
Writing $Y(\xi_j^{x,k}\otimes \id)=\sum_{i}\xi^{x,k}_i\otimes \chi(u^{x,k}_{ij})$,
we have
\begin{equation}\label{calcanal}
Y^*(\id\otimes \chi)\alpha_U(a)Y(\xi^{z,m}_n\otimes 1)=\sum_{x,k,i,j,q}\xi^{x,k}_i\langle\xi_j^{x,k},a\xi^{z,m}_q\rangle\otimes \sum_{s,t}(\chi(u^{x,k}_{si}))^*\chi(u^{x,k}_{sj}(u^{z,m}_{tq})^*)\chi(u^{z,m}_{tn}).
\end{equation} 
Note now that by equation \eqref{eqchi}, 

$$\chi(u^{x,k}_{sj}(u^{z,m}_{tq})^*)=\chi(u^{x,k}_{sv})\chi((u^{z,m}_{tr})^*)\sigma(u^{x}_{vj},(u^{z}_{rq})^*)$$ and by unitarity of the $u^x_{ij}$'s and the $\chi(u^x_{ij})$'s and theorem \ref{multkappa}, we get 
\begin{equation}
\sum_{s,t}(\chi(u^{x,k}_{si}))^*\chi(u^{x,k}_{sj}(u^{z,m}_{tq})^*)\chi(u^{z,m}_{tn})=\sum_{t,r}\sigma^{-1}((u^{z,m}_{tr})^*,u^{z,m}_{tn})\sigma(u^{x,k}_{ij},(u^{z,m}_{rq})^*)
\end{equation} which implies
$$Y^*(\id\otimes \chi)\alpha_U(a)Y(\xi^{z,m}_n\otimes 1)= \sum_{x,k,i,j,q}\xi^{x,k}_i\langle\xi_j^{x,k},a\xi^{z,m}_q\rangle\sum_{t,r}\sigma^{-1}((u^{z,m}_{tr})^*,u^{z,m}_{tn})\sigma(u^{x,k}_{ij},(u^{z,m}_{rq})^*).$$ We can conclude that
$$\phi(a\cdot_\sigma \xi)=(\id\otimes \chi)\alpha(a)\phi(\xi)$$ and with this, the proof of theorem \ref{maincocycle} is completed.
\end{proof}

\section{Constructing a non-dimension-preserving example}
\label{cha-inducingmoneq}

In this section, we will construct an example of a monoidal deformation coming from a non-dimension-preserving monoidal equivalence. We will use the spectral triple on the Podle\'s spheres (\cite{Podles1987}) defined in \cite{Dabrowski2006} and $SU_q(2)$, which acts on it in the appropriate way. 
%To have monoidal equivalences on $SO_q(3)$, we will use that $C(SO_q(3))$ is a Woronowicz $C^*$-subalgebra of $SU_q(2)$: we first describe a method to induce a monoidal equivalence between two Woronowicz $C^*$-subalgebras of two monoidal equivalent compact quantum groups and after that, we investigate monoidal equivalences on $SU_q(2)$.

\subsection{Monoidal equivalences on $SU_q(2)$}\label{subsec-monoidalSUQ2}
We look at orthogonal quantum groups and $SU_q(2)$ in particular.
\begin{definition}[\cite{VanDaele1996b}]\label{defAoF}
Let  $n\in \mathbb{N}$ and $F \in \GL(n,\mathbb{C})$ with $F\overline{F}=cI_n \in \mathbb{R}I_n$. Then $A_o(F)$ is defined as the universal quantum group generated by the coefficients of the matrix $U\in M_n(A_o(F))$ with relations

\begin{itemize}
\item $U$ is a unitary and
\item $U = F\overline{U}F^{-1}$.
\end{itemize}
Moreover, $A_o(F)=(C(A_o(F)),U)$ is a compact matrix quantum group (as defined in \cite{Woronowicz1987}). They are called universal orthogonal quantum groups.
\end{definition}
As the matrices $F$ are not in one to one correspondence with the universal quantum groups (i.e. different $F$'s can define the same universal quantum groups), it is necessary (but not so hard) to classify the quantum groups $A_o(F)$. This has been done in \cite{Bichon2005}.

\begin{proposition}
For $F_1,F_2$ matrices in $\GL(n,\mathbb{C})$ with $F_i\overline{F}_i=\pm 1$, we say 
$$F_1\sim F_2  \text{ if there exists a unitary }  v\in U(n) \text{ such that } F_1=vF_2v^T.$$ Then 
$$A_o(F_1)\cong A_o(F_2) \text{ if and only if } F_1\sim F_2.$$
\end{proposition}
Therefore, we will describe a fundamental domain for $\sim$ as is done in \cite{Bichon2005}.
\begin{proposition}
A fundamental domain of $\sim$ is given by the following classes of matrices:
\begin{itemize}
\item $\left(\begin{array}{ccc}0 & D(\lambda_1,\ldots,\lambda_k) & 0 \\D(\lambda_1,\ldots,\lambda_k)^{-1} & 0 & 0 \\0 & 0 & 1_{n-2k}\end{array}\right)$ with $k,n\in \mathbb N, 2k\leq n, 0< \lambda_1\leq \ldots\leq \lambda_k<1$
\item $\left(\begin{array}{cc}0 & D(\lambda_1,\ldots,\lambda_{n/2}) \\-D(\lambda_1,\ldots,\lambda_{n/2})^{-1} & 0\end{array}\right)$ with $0< \lambda_1\leq \ldots\leq \lambda_{n/2}\leq1,n\in \mathbb N$ even.
\end{itemize}

\end{proposition}
\begin{remark}\label{remFq}
Note that for $F\in \GL(2,\mathbb{C})$, up to equivalence, there only exists matrices of the form 
$$ F_q=\left(\begin{array}{cc}0 & |q|^{1/2} \\-\sgn(q)|q|^{-1/2} & 0\end{array}\right)$$ for $q\in [-1,1]\setminus \{0\}$.
\end{remark}

\begin{definition}[\cite{Woronowicz1987,Woronowicz1987a}]
Let $q\in [-1,1], q\neq 0$. Let $A$ be the universal unital $C^*$-algebra generated by two elements $\alpha,\gamma$ satisfying relations such that  $U=\left(\begin{array}{cc}\alpha & -q\gamma^* \\\gamma & \alpha^*\end{array}\right)\in M_2(A)$ is a unitary matrix. With coproduct $\Delta(U_{ij})=\sum_k U_{ik}\otimes U_{kj}$, $SU_q(2)=(A,\Delta)$ is a compact quantum group. 
%On the generators $\varepsilon(U_{ij})=\delta_{ij}$ and antipode $S(U_{ij})=(U_{ji})^*$, $\mathcal C(SU_q(2))$is a Hopf $^*$-algebra.
\end{definition}
%Moreover, we have a complete description of basis elements of $SU_q(2)$.
%
%\begin{proposition}\label{basisSUq}
%The elements of the form $\alpha^i\gamma^j(\gamma^k)^*$ and $(\alpha^*)^l\gamma^m(\gamma^*)^n$, $i,j,k,m,n\in \mathbb{N},l\in\mathbb{N}\setminus \{0\}$ form a basis of $\SUq$
%\end{proposition}
%
\begin{proposition}
With $F_q$ defined in remark \ref{remFq}, we have $A_o(F_q)\cong SU_q(2)$.
\end{proposition}
Note that this last statement indeed implies that the only orthogonal quantum groups coming from matrices of dimension $2$, are the quantized versions of $SU(2)$. 

%\subsection{$SO_q(3)$ as quantum quotient group of $SU_q(2)$.}
%\section{Monoidal equivalences on $SO_q(3)$}
%\subsection{Monoidal equivalences between $A_o(F)$'s}
We state some results obtained by de Rijdt et al. in \cite{Bichon2005} (Corollary 5.4 and Theorem 5.5).

\begin{theorem}\label{moneqAoF}
Let $F_1\in GL(n_1,\mathbb{C})$ with $F_1\overline{F}_1=c_11$, $c_1\in \mathbb{R}$. Then

\begin{itemize}
\item a compact quantum group $\mathbb{G}$ is monoidally equivalent with $A_o(F_1)$ if and only if there exist a $F_2\in GL(n_2,\mathbb{C})$ with $F_2\overline{F}_2=c_21, c_2\in \mathbb{R}$ and $\frac{c_1}{\Tr(F_1^*F_1)}=\frac{c_2}{\Tr(F_2^*F_2)}$ such that $\mathbb{G}\cong A_o(F_2)$.
\item in this case, denote by $\mathcal O(A_o(F_1,F_2))$ the $^*$-algebra generated by the coefficients of $Y\in M_{n_2,n_1}(\mathbb{C})\otimes \mathcal O(A_o(F_1,F_2))$ with relations $$ Y \text{ is unitary } \quad \text{ and }\quad Y=(F_2\otimes 1)\overline{Y}(F_1^{-1}\otimes 1),$$ then $\mathcal O(A_o(F_1,F_2))\neq 0$ is the ($A_o(F_1)$-$A_o(F_2)$)-bi-Galois object with coactions $\beta_1$ of $\mathcal O(A_o(F_1))$ and $\beta_2$ of $\mathcal O(A_o(F_2))$ such that 
$$(\id\otimes \beta_1)(Y)=Y_{12}(U_1)_{13} \quad \text{ and } \quad (\id\otimes \beta_2)(Y)=(U_2)_{12} Y_{13}$$ where the $U_i$ are the unitary representations of $A_o(F_i)$, which matrix coefficients generate the quantum groups.
\item the monoidal equivalence preserves the dimensions if and only if $n_2=n_1$. In this case, we denote the unitary 2-cocycle by $\Omega(F_2)$. The $\Omega(F_2)$ describe up to equivalence all unitary 2-cocycles on the dual of $A_o(F_1)$.
\end{itemize}
\end{theorem}

\begin{remark}\label{nondimpreserving}
In \cite{Banica1996} Banica shows that the irreducible representations of $A_o(F)$ can be labeled by $\mathbb N$ (say $r_k$, $k\in \mathbb N$). Moreover, for $\dim(F)=n$, he states that $\dim(r_k)=(x^{k+1}-y^{k+1})/(x-y)$ where $x$ and $y$ are solutions of $X^2-nX+1=0$ for $n\geq 3$ and $\dim(r_k)=k+1$ for $n=2$. Hence, it is easy to show by induction that if $\varphi$ is a monoidal equivalence between $SU_q(2)$ and $A_o(F)$ with $\dim(F)\geq 4$, then $\dim(\varphi(r_k))>\dim(r_k)=k+1$ for every irreducible representation $r_k$ with $k\geq 1$.
\end{remark}

Moreover, looking at the concrete orthogonal quantum group $SU_q(2)$,  it is possible to classify all compact quantum groups which are monoidally equivalent with $SU_q(2)$: indeed applying the result of the last paragraph to the specific situation of $F=F_q$, we know exactly what the quantum groups are which are monoidal equivalent with $SU_q(2)$. 

\begin{proposition}[\cite{Bichon2005}]\label{clasmoneqSUq}
Let $0<q\leq 1$. For every even natural number $n$ with $2\leq n\leq q+1/q$, there exists a monoidal equivalence on $SU_q(2)$ such that the multiplicity of the fundamental representation is $n$. Concretely, $SU_q(2)\sim_{mon} A_o(F)$ with 
$F=\left(\begin{array}{cc}0 & D(\lambda_1,\ldots,\lambda_{n/2}) \\-D(\lambda_1,\ldots,\lambda_{n/2})^{-1} & 0\end{array}\right)$ where $0< \lambda_1\leq \ldots\leq \lambda_{n/2}\leq1$ and $\sum_{i=1}^{n/2}\frac{1}{\lambda_i^2}+\lambda_i^2=q+1/q$.\\Ê\\
Let $0>q\geq -1$. Then for every natural number $n$ with $2\leq n\leq |q+1/q|$, there exists a monoidal equivalence on $SU_q(2)$ such that the multiplicity of the fundamental representation is $n$. Concretely, $SU_q(2)\sim_{mon} A_o(F)$ with 
$F=\left(\begin{array}{ccc}0 & D(\lambda_1,\ldots,\lambda_{k})&0 \\D(\lambda_1,\ldots,\lambda_{k})^{-1} & 0&0\\0&0&1_{n-2k}\end{array}\right)$ where $k\in \mathbb N, 2k\leq n,0< \lambda_1\leq \ldots\leq \lambda_{k}<1$ and $\sum_{i=1}^{k}\frac{1}{\lambda_i^2}+\lambda_i^2+n-2k=|q+1/q|$.\\Ê\\
\end{proposition}

%\begin{remark}
%Applying theorem \ref{moneqAoF} to the case of $SU_q(2)$ we see also that all unitary 2-cocycles on the dual of $SU_q(2)$ are coboundaries. 
%\end{remark}

\subsection{Monoidal deformation of the Podle\'s sphere}\label{subsec-mondefpodles}
In section \ref{cha-cocycle}, we proved that our monoidal deformation of spectral triples is a generalization of the cocycle deformation, developed in \cite{Goswami2014}. In this subsection, we will give a concrete example to prove that our construction is a genuine generalization: we will construct a monoidal deformation of the Podle\'s sphere (with spectral triple of Dabrowski, Landi, Wagner and D'Andrea \cite{Dabrowski2006}) which is not a 2-cocycle deformation. First we recapitulate the definition of the Podles sphere $S^{2}_{q,c}$ and the spectral triple on it. Then we will use the results of subsection \ref{subsec-monoidalSUQ2} to apply the construction of section \ref{sec-mondef}.

\subsubsection{The Podle\'s sphere, its spectral triple and its quantum isometry group}
The Podles sphere was initially constructed by Podle\'s in \cite{Podles1995} as follows. Let $q\in (0,1)$ and $t\in (0,1)$, hence $c=t^{-1}-t>0$. We define $\mathcal O(S^2_{q,c})$ to be the $^*$-algebra generated by elements $A,B$ which satisfy the relations

\begin{eqnarray*}
A^*=A, &\qquad&AB=q^{-2}BA,\\
B^*B=A-A^2+c1,&\qquad& BB^*=q^2A-q^4A^2+c1.\\
\end{eqnarray*}
You can see that for $q=1$, we have $A^*=A, AB=BA, B^*B=BB^*=A-A^2+c1$ and this is the classical sphere: putting $A=z+1/2, B=x+iy,r^2=c+1/4$, we indeed have
$$ x^2+y^2+z^2=B^*B+A^2-A+1/4=c+1/4=r^2.$$
The associated quantum space is called the Podles sphere $S^2_{q,c}$.

Note first that for $q\in(0,1)$, setting 
$$x_0=t(1-(1+q^2)A), x_{-1}=\frac{t(1+q^2)^{\frac{1}{2}}}{q}B, x_1=-t(1+q^2)^{\frac{1}{2}}B^*,$$ we see that the definition in \cite{Dabrowski2006} with $\{x_0,x_{-1},x_{1}\}$ is equivalent to the original definition of Podle\'s given above.
Moreover, defining 
$$\tilde A= \frac{1+t^{-1}q\gamma^*\alpha-t^{-1}\rho(1-(1+q^2)\gamma^*\gamma)+t^{-1}\gamma\alpha^*}{1+q^2}$$
$$\tilde B=\frac{q\alpha^2+\rho(1+q^2)\alpha\gamma-q^2\gamma^2}{t(1+q^2)},$$ where $\rho^2=\frac{q^2t^2}{(q^2+1)^2(1-t)}$,
one can prove that the unital $^*$-subalgebra of $C(SU_q(2))$ generated by $\tilde A$ and $\tilde B$ is isomorphic to $\mathcal O(S^2_{q,c})$ where $c=t^{-1}-t$, sending $A$ to $\tilde A$ and $B$ to $\tilde B$.\\
Doing as above, we have 3 equivalent descriptions of the Podles sphere.\\ \\
The spectral triple on $S^2_{q,c}$ we will use, is the spectral triple developed by Dabrowski, D'Andrea, Landi and Wagner in \cite{Dabrowski2006}.
The spectral triple uses the representation theory of $SU_q(2)$ described by Banica in \cite{Banica1996}. To be compatible with \cite{Dabrowski2006}, we use their notation. For each $n$ in $\{0,1/2,1,\ldots\}$, there exists a unique irreducible representation $D^n$ ($r_{2n}$ in Banica's notation) of dimension $2n+1$. For example , we have 
$$ D^{1/2}= \left(\begin{array}{cc}\alpha & -q\gamma^* \\\gamma & \alpha^*\end{array}\right)$$ and
$$D^1=\left(\begin{array}{ccc}{\alpha^*}^2 & -(q^2+1)\alpha^*\gamma & -q\gamma^2 \\\gamma^*\alpha^* & 1-(q^2+1)\gamma^*\gamma & \alpha\gamma \\-q{\gamma^*}^2 & -(q^2+1)\gamma^*\alpha & \alpha^2\end{array}\right).$$
Denoting $d^n_{k,l}$ to be the $k,l$-matrix coefficient of $D^n$, one can prove that $$\{d^n_{k,l}\mid n=0,\frac{1}{2},1,\ldots;k,l=-n,-n+1,\ldots,n-1,n\}$$ form an orthogonal basis of $\mathcal K=L^2(SU_q(2),h)$, the GNS-space corresponding to the Haar state $h$ of $SU_q(2)$. Moreover we will denote $e^n_{k,l}$ the multiples of $d^n_{k,l}$ such that the $\{e^n_{k,l}\}$ form an orthonormal basis of $\mathcal K$.

Furthermore, defining the new Hilbert space 
$$\mathcal{H}:=\big[e^n_{\pm\frac{1}{2},l}\mid n=\frac{1}{2},\frac{3}{2},\ldots;l=-n,-n+1,\ldots,n-1,n \big ],$$ one can prove that $\tilde A$ and $\tilde B$, as defined above, leave $\mathcal H$ invariant and we have a faithful $^*$-morphism $\pi:\mathcal O(S^2_{q,c})\to B(\mathcal H):A\mapsto \tilde A_{|_{\mathcal H}}, B\mapsto \tilde B_{|_{\mathcal H}}$, which makes it possible to identify $\mathcal O(S^2_{q,c})$ with its image.\\

Finally, we can define an appropriate Dirac operator by setting 
$$D(e^n_{\pm\frac{1}{2},l})=(c_1n+c_2)e^n_{\mp\frac{1}{2},l}$$ where $c_1,c_2\in \mathbb R, c_1\neq 0$ are arbitrary constants. \\
In \cite{Dabrowski2006} the authors prove that $(\mathcal O(S^2_{q,c}),\Hi, D)$ constitutes a well defined spectral triple.  As $$\Delta_{SU_q(2)}(e^n_{\pm\frac{1}{2},l})=\sum_{k=-n,-n+1,\ldots,n}e^n_{\pm\frac{1}{2},k}\otimes e^n_{k,l}$$ it is easy to see that $\Delta_{SU_q(2)}$ induces a unitary representation $U$ of $SU_q(2)$ on $\Hi$. By \cite{Dabrowski2006} the spectral triple is equivariant with respect to  this representation and hence, $SU_q(2)$ acts algebraically and by orientation-preserving isometries on  $(\mathcal O(S^2_{q,c}),\Hi, D)$. We will use this representation and the monoidal equivalences of subsection \ref{subsec-monoidalSUQ2} to deform  this spectral triple.

\subsubsection{Monoidal deformation of the Podle\'s sphere}
To conclude this section, we construct a non-dimension-preserving example. Now we know that there is a well defined spectral triple $(\mathcal O(S^2_{q,c}),\Hi,D)$ on which $SU_q(2)$ acts algebraically and by orientation-preserving isometries. Furthermore, we know from proposition \ref{clasmoneqSUq} what the monoidal equivalences of $SU_q(2)$ are and %by theorem \ref{indmoneq}, what those are of $SO_q(3)$. Finally, 
we know that those monoidal equivalences are non-dimension-preserving by remark \ref{nondimpreserving}. Putting all this together, we can apply the construction described in section \ref{sec-mondef} to get the following theorem.
\begin{theorem}\label{thm-defpodles}
Let $q\in [-1,1] \setminus\{0\}$ and $n$ a natural number with $3\leq n\leq |q+1/q|$.\\
If $q> 0$ and $n$ is even, let $\lambda_1,\ldots,\lambda_{n/2}$ be strict positive real numbers not bigger than 1 such that $\lambda_1^2+\ldots+\lambda_{n/2}^2+1/\lambda_1^2+\ldots+1/\lambda_{n/2}^2=q+1/q$ and define $F$ to be the $n$ by $n$ matrix $$F=\left(\begin{array}{cc}0 & D(\lambda_1,\ldots,\lambda_{n/2}) \\-D(\lambda_1,\ldots,\lambda_{n/2})^{-1} & 0\end{array}\right).$$ 
If $0>q$, let $k$ be a natural number $k\leq n/2$ and $\lambda_1,\ldots,\lambda_{k}$ be strict positive real numbers such that $0< \lambda_1\leq \ldots\leq \lambda_{k}<1$ and $\sum_{i=1}^{k}\frac{1}{\lambda_i^2}+\lambda_i^2+n-2k=|q+1/q|$ and define $F$ to be the $n$ by $n$ matrix
$$F=\left(\begin{array}{ccc}0 & D(\lambda_1,\ldots,\lambda_{k})&0 \\D(\lambda_1,\ldots,\lambda_{k})^{-1} & 0&0\\0&0&1_{n-2k}\end{array}\right).$$

With $F$ defined as above, there exists a non-dimension-preserving monoidal equivalence $\varphi$ from $SU_q(2)$ to $A_o(F)$ (introduced in definition \ref{defAoF}). Denoting by $\mathcal O(A_o(F_q,F))$ the algebra constructed in \ref{moneqAoF}, $\mathcal O(A_o(F_q,F))$ is the associated bi-Galois object and the following triplet is a spectral triple:
$$\big(\mathcal{O}(S_{q,c}^2)\boksdot[\mathcal{O}(SU_q(2))] \mathcal O(A_o(F_q,F)),\quad\Hi\bokstimes[C(SU_q(2))] L^2 \big(\mathcal O(A_o(F_q,F))\big), \quad\tilde D\big).$$ Moreover $A_o(F)$ acts algebraically and by orientation-preserving isometries on the new spectral triple. As $\varphi$ is non-dimension-preserving, it is not a 2-cocycle deformation \`a la Goswami-Joardar \cite{Goswami2014}.
\end{theorem}

\section{Deformation of the quantum isometry group}
The goal of this last section is to prove that the deformation (in the sense of theorem \ref{mainthm}) of the quantum isometry group of a spectral triple (defined by Bhowamick and Goswami) is the quantum isometry group of the deformed spectral triple. We start by recalling some concepts and results of \cite{Bhowmick2009}.

\begin{definition}[Definition 2.7 in \cite{Bhowmick2009}]\label{defRtwistedST}
An $R$-twisted spectral triple (of compact type) is given by a triple $(\A,\Hi,D)$ and an operator $R$ on $\Hi$ where
\begin{enumerate}
\item $(\A,\Hi,D)$ is a compact spectral triple,
\item $R$ is a positive (possibly unbounded) invertible operator such that $R$ commutes with $D$
%\item For all $s\in \mathbb R$, the map $a\mapsto \sigma_s(a)=R^{-s}aR^s$ gives an automorphism of $\A$ (not necessarily $*$-preserving) satisfying $\sum_{s\in [-n,n]}\|\sigma_s(a)\|<\infty$ for all positive integers $n$.
\end{enumerate}
\end{definition}
\begin{remark}
We note that in Definition 2.7 in \cite{Bhowmick2009}, there is a third condition in the definition of $R$-twisted spectral triple. However in remark 2.11 of \cite{Bhowmick2009}, the authors state that this third condition is not necessary. Therefore, we gave the definition above.
\end{remark}
Such an operator $R$ is linked with the preservation of a non-commutative analogue of a volume form.

\begin{definition}[\cite{Bhowmick2009}]
Let $R$ be a positive invertible operator and $(\A,\Hi,D)$ an $R$-twisted spectral triple. Then a compact quantum group $\mathbb G$ acting on $(\A,\Hi,D)$ by orientation-preserving isometries is said to preserve the $R$-twisted volume if one has
$$(\tau_R\otimes \id)(\alpha_U(x))=\tau_R(x)1_{\G}$$ 
for all $x\in \mathcal E_D$, where $\tau_R(x)=\Tr(Rx)$ and where $\mathcal E_D$ is the $^*$-subalgebra of $B(\Hi)$ generated by the rank-one operators of the form $\eta\xi^*$, $\eta,\xi$ eigenvectors of $D$.
\end{definition}

In what follows we will denote by $\mathcal Q_R(\A,\Hi,D)$ (or just $\mathcal Q_R$) the category of all compact quantum groups acting by $R$-twisted volume- and orientation-preserving isometries with as morphisms the morphisms of quantum groups which are compatible with the representations on $\Hi$.\\ \\
Moreover, one can prove (as is done in \cite{Goswami2004}) that for every compact quantum group acting by orientation-presering isometries, there exists an operator $R$ such that the quantum group is an elements of $\Q_R$. \\Ê\\
Now Goswami and Bhowmick proved in \cite{Bhowmick2009} that there exists a universal object in $\Q(\A,\Hi,D)$.

\begin{theorem}[Theorem 2.14 in \cite{Bhowmick2009}]\label{QISO0existence}
For any $R$-twisted spectral triple $(\A,\Hi,D)$ there exists a universal (initial) object $(\QISORnul,U_0)$ in the category $\Q_R$. The representation is faithful.
\end{theorem}
For notational convenience, we will write $\QISORnulb$ if there is no confusion possible about the spectral triple.
However, in general $\alpha_{U_0}$ may not be faithful even if  $U_0$ is so. Therefore one has the following definition.
\begin{definition}[Definition 2.16 in \cite{Bhowmick2009}]\label{QISO}
Let $\C=C^*(\{(f\otimes \id)\alpha_{U_0}(a)\mid a\in \A,f\in \A^*\})$ be the $C^*$-subalgebra of $C(\QISORnulb)$ generated by elements of the form $(f\otimes \id)\alpha_{U_0}(a), a\in \A$.

Then $\C$ is a Woronowicz $C^*$-subalgebra of $\QISORnulb$ and the compact quantum group $$\QISOR=(\C,{\Delta_{\QISORnulb}}_{|_{\C}})$$ is called the quantum group of $R$-twisted volume- and orientation-preserving isometries or simply quantum isometry group.
\end{definition}
In subsection \ref{subse-defQISO}, we will prove that if $(\A,\Hi,D)$ is an $R$-twisted spectral triple and $$\varphi:\QISOR\to\mathbb G_2$$ is a monoidal equivalence, then there exists an operator $\tilde R$ such that $(\tilde \A,\tilde \Hi,\tilde D)$ is an $\tilde R$-spectral triple and $\mathbb G_2=\QISORtilde$. %In subsection \ref{subsec-defQISO} we extend this result: if $(\A,\Hi,D)$ is an $R$-twisted spectral triple and if we have a monoidal equivalence $\varphi:\QISORnul\to\mathbb G_2$ then we also have monoidal equivalence $\varphi':\QISOR\to\mathbb H_2$ such that $\mathbb H_2=\QISORtilde$. 
But before we do that, we describe, given a monoidal equivalence $\varphi: \mathbb G_1\to \mathbb G_2$, how to construct a monoidal equivalence between certain Woronowicz-$C^*$-subalgebras (subsection \ref{subsec-inducing}) resp.  quantum supergroups (subsection \ref{subsec-inducingonsupergroups}) of $\mathbb G_1$ and $\mathbb G_2$.

\subsection{Inducing monoidal equivalences on Woronowicz-$C^*$-subalgebras}\label{subsec-inducing}

\begin{definition}[\cite{Baaj1993}]
Let $\mathbb{G}=(\G,\Delta)$ be a compact quantum group and $A$ a $C^*$-subalgebra of $\G$ such that $\Delta(A)\subset A\otimes A$ and $[\Delta_{|_A}(A)(A\otimes 1)]=A\otimes A=[\Delta_{|_A}(A)(1\otimes A)]$.
%\begin{itemize}
%\item 
Then $A$ is called a Woronowicz $C^*$-subalgebra. We will write $\mathbb A =(A,\Delta_{|_A})$ for the quantum group. 
%\item If moreover, $A=\GN$ for some quantum subgroup $\mathbb H$ and if one of the conditions in proposition \ref{propnsubg} is satisfied, we call $\mathbb{H}$ a normal quantum subgroup of $\mathbb{G}$ and $\mathbb{G}/\mathbb{H}$ a quantum quotient group.
%w\end{itemize}
\end{definition}

It is good to remark that the notion of compact quantum quotient group introduced in \cite{Wang1995} is a special case of a Woronowicz $C^*$-subalgebra. However it is still unknown whether all Woronowicz $C^*$-subalgebras are compact quantum quotient groups. 

%\subsection{representations of quotients of compact quantum groups}
In this section let $\mathbb G=(\G,\Delta)$ be a CQG and $A$ a Woronowicz $C^*$-subalgebra of $\mathbb G$.
In order to define a unitary fiber functor on $\mathbb A$, it is good to examine its representations. It is easy to see that every representation $U$ of $\mathbb A$ on a Hilbert space $\Hi$ is a representation of $\mathbb G$ and that every representation $V$ of $\mathbb G$ is a representation of $A$ if and only if $V\in \mathcal M(\mathcal K(\Hi)\otimes A)$. To distinguish, we will write $U_{\mathbb G}$ for a representation $U$ of $\mathbb A$ seen as representation of $\mathbb G$.
Moreover, we have the following proposition
\begin{proposition}
Let $U$ be a unitary representation of $\mathbb A$. Then $U$ is irreducible if and only if $U_{\mathbb G}$ is irreducible.
\end{proposition}

\begin{proof}
We know that $U$ resp. $U_{\mathbb G}$ is irreducible if and only if $\Mor(U,U)=\{T\in B(\Hi)|(T\otimes \id)U=U(T\otimes \id)\}$ resp. $\Mor(U_{\mathbb G},U_{\mathbb G})$ equals $\mathbb C 1_{B(\Hi)}$. As it is directly clear that $\Mor(U,U)=\Mor(U_{\mathbb G},U_{\mathbb G})$, the proposition is proved.
\end{proof}
Analogously as before, we will write $x_{\mathbb G}$ if we look at the equivalence class $x\in \Irred(\mathbb A)$ seen as equivalence class in $\Irred(\mathbb G)$. Using this proposition, the unitary fiber functor is easily made: let $\mathbb{G}_1$ be a compact quantum group and $\varphi: \mathbb{G}_1\to \mathbb{G}_2$ a monoidal equivalence between them. Suppose moreover that $A_1$ is a Woronowicz subalgebra of $\mathbb{G}_1$. Then we can construct a unitary fiber functor on $\mathbb A_1=(A_1,\Delta_{|_{A_1}})$ by restricting $\varphi$ to the representations of $\mathbb A$ and proof it is a monoidal equivalence between $\mathbb A_1$ and a compact quantum group $\mathbb A_2$ such that $C(\mathbb A_2)$ is a Woronowciz $C^*$-algebra of $\mathbb G_2$.

\begin{proposition}\label{indunifibfunct}
Let $\mathbb{G}_1$ be a compact quantum group, $ A_1$ a Woronowicz $C^*$-subalgebra of $\mathbb{G}_1$ and $\psi$ a unitary fiber functor on $\mathbb G_1$. Then there exists a unitary fiber functor $\psi'$ on $\mathbb A_1=(A_1,{\Delta_1}_{|_{A_1}})$  such that $\psi'(x)=\psi(x_{\mathbb G_1})$ for all $x\in \Irred(\mathbb A_1)$.
\end{proposition}

\begin{proof}
Let $x\in \Irred(\mathbb A_1)$. Define $\Hi_{\psi'(x)}$ to be $\Hi_{\psi(x_\mathbb{G})}$ and $\psi'(S)=\psi(S)$ for every $S\in \Mor(y_1\otimes\ldots\otimes y_k,x_1\otimes\ldots\otimes x_r), y_1,\ldots,y_k,x_1,\ldots,x_r\in \Irred(\mathbb{A}_1)$. As $\psi$ is a unitary fiber functor, $\psi'$ will satisfy all the necessary conditions to be a unitary fiber functor as well. 
\end{proof}
Denoting by $\varphi :\mathbb G_1 \to \mathbb G_2$ the monoidal equivalence associated to $\psi$, we can see $C(\mathbb{G}_2)$ as the $C^*$-algebra generated (as vector space) by the coefficients of the $U^{\varphi(x)}, x\in \Irred(\mathbb{G}_1)$. Now we can define $A_2$ as the $C^*$-algebra generated (as vector space) by the coefficients of the $U^{\varphi(x_{\mathbb G_1})}, x\in \Irred(\mathbb A_1)$. Equivalently, $$A_2=[ (\omega\otimes \id)U^{\varphi(x_{\mathbb G_1})}| x\in \Irred(\mathbb A_1)]$$ and we also write 
$$\mathcal A_2=\langle(\omega\otimes \id)U^{\varphi(x_{\mathbb G_1})}| x\in \Irred(\mathbb A_1)\rangle$$

Now it is clear that $\psi'$ induces a monoidal equivalence $\varphi'$ between $\mathbb A_1$ and a compact quantum group with algebra $A_2$.

\begin{theorem}\label{indmoneq}
With the map $\Delta'_2=\Delta_{2|_{A_2}}$, $\mathbb{A}_2=(A_2,\Delta'_2)$ is a compact quantum group. Moreover the monoidal equivalence $\varphi'$, induced by $\psi$ is an equivalence between $\mathbb A_1$ and $\mathbb{A}_2$.
\end{theorem}
\begin{proof}
Written differently, $A_2$ is the closed linear span of the elements $u^{\varphi(x_{\mathbb G_1})}_{ij}, x\in \Irred(\mathbb{A}_1)$. For $\Delta'_2$ defined as above, we get: 
$$\Delta'_2(u^{\varphi(x_{\mathbb G_1})}_{ij})=\sum_k u^{\varphi(x_{\mathbb G_1})}_{ik}\otimes u^{\varphi(x_{\mathbb G_1})}_{kj}$$ and as $x\in\Irred(\mathbb{A}_1)$, we see that $\Delta'_2(A_2)\subset A_2\otimes A_2$. Now denote by $\varepsilon'$ and $S'$ the restrictions of the counit $\varepsilon $ and antipode $S$ of $\mathbb{G}_2$ defined on $\OGb$ to $\A_2$. Then $\A_2=\langle (\omega\otimes \id)U^{\varphi(x_{\mathbb G_1})}| x\in \Irred(\mathbb{A}_1)\rangle=\mathcal{O}(\mathbb{A}_2)$ is a Hopf $^*$-algebra which is dense in $A_2$. This proves that $\mathbb A_2=(A_2,\Delta'_2)$ is indeed a compact quantum group. By construction of $\varphi'$, it is evident that it is a monoidal equivalence between $\mathbb{A}_1$ and $\mathbb{A}_2$.
\end{proof}
Before we go the next paragraph, we want to explore how the $\mathbb G_1-\mathbb G_2$-bi-Galois object behaves with respect to the $\mathbb A_1-\mathbb A_2$-bi-Galois object .

\begin{theorem}\label{thm-galoisobjectquotient}
Let $\mathbb{G}_1,\mathbb{G}_2,\mathbb{A}$ be compact quantum groups such that $C(\mathbb A)$ is a Woronowicz $C^*$-subalgebra of $\Ga$ and such that $\varphi: \mathbb{G}_1\to\mathbb{G}_2$ is a monoidal equivalence. Let $\B$ be the $\mathbb{G}_1-\mathbb{G}_2$-bi-Galois object with $\beta_1:\B\to \OGa\odot\B $. Let $\varphi'$ be the monoidal equivalence between $\mathbb{A}_1$ and $\mathbb{A}_2$ as defined above and define $\B'$ to be the $\mathbb{A}_1-\mathbb{A}_2$-bi-Galois object with $\gamma_1:\B'\to \mathcal{O}(\mathbb A_1)\odot \B'$. Then we have
$$\B'=\{b\in \B| \beta_1(b)\in  \mathcal{O}(\mathbb A_1)\odot\B\}$$ and $\gamma_1=\beta_{1|_{\B'}}$.
\end{theorem}

\begin{proof}
From the original proof of theorem \ref{monoidalGalois} (which is theorem 3.9 in \cite{Bichon2005}), we know that $\B'=\oplus_{x\in \Irred(\mathbb A_1)}B(\Hi_{\varphi(x)},\Hi_x)^*$ and $\B=\oplus_{x\in \Irred(\mathbb G_1)}B(\Hi_{\varphi(x)},\Hi_x)^*$. Hence $\B'\subset \B$. Also, $X^x\in B(\Hi_{\varphi(x)},\Hi_x)\odot \B$ is defined such that $(\omega_x\otimes \id)(X^x)=(\delta_{x,y}\omega_x)_{y\in \Irred(\mathbb G_1)}$ for all $\omega_x\in B(\Hi_{\varphi(x)},\Hi_x)^*$. By definition, we see that for $x\in  \Irred(\mathbb A_1)$, $X^x=X^{x_{\mathbb G_1}}$. As $\beta_1$ resp. $\gamma_1$ are defined by $(\id\otimes \beta_1)(X^x)=U_{12}^xX^x_{13}$ ($x\in \Irred(\mathbb G_1)$) resp. $(\id\otimes \gamma_1)(X^x)=U_{12}^xX^x_{13}$ ($x\in \Irred(\mathbb A_1))$, it follows directly that $\gamma_1=(\beta_1)_{|_{\B'}}$. Moreover, if $x\in \Irred(\mathbb A_1)$, $U_{12}^xX^x_{13}\in B(\Hi_{\varphi(x)},\Hi_x)\odot\mathcal{O}(\mathbb A_1)\odot\B$ and hence for $b\in \B', \beta_1(b)\in  \mathcal{O}(\mathbb A_1)\odot\B$. If  $x\in \Irred(\mathbb G_1)$ but $x\notin \Irred(\mathbb A_1)$, $U_{12}^xX^x_{13}\notin B(\Hi_{\varphi(x)},\Hi_x)\odot\mathcal{O}(\mathbb A_1)\odot\B$ and hence for $b\in \B$ but $b\notin \B', \beta_1(b)\notin  \mathcal{O}(\mathbb A_1)\odot\B$. This concludes the proof.
%is generated as vector space by the matrix coefficients of the $X^y$, $y\in \Irred(\mathbb A_2)$. Moreover, $(\id\otimes \beta'_1)X^y=U^{\varphi'(y)}_{12}X^y_{13}$ for all $y\in \Irred(\mathbb G_2)$, which implies that $(\id\otimes \beta'_1)X^y\in B(\Hi_y)\odot \mathcal C(\mathbb A_1)\odot \B$ if and only if $y\in \Irred(\mathbb A_2)$. Therefore  $\beta_1'(\B')\in \mathcal C(\mathbb A_1)\odot\B$. 
%Furthermore, for $b\in \B\setminus \B'$, $\beta_1(b)\notin \mathcal C(\mathbb A_1)$ as $(\id\otimes \beta_1)X^x=U^{x}_{12}X^x_{13}$ for all $x\in \Irred(\mathbb G_1)$ and hence $\B'=\{b\in \B| \beta_1(b)\in  \mathcal C(\mathbb A_1)\odot\B\}$. 
\end{proof}
%The above theorem reaffirms the equivalence between the algebraic description given in chapter \ref{cha-algdef} theorem \ref{indgalobj} and the analytic description, given in this chapter and theorem \ref{indmoneq}.

\begin{remark}
In the special case of compact quantum quotient groups, a compact quantum quotient group of $\mathbb G_1$ will be monoidally equivalent with a compact quantum group which has as algebra a Woronowicz $C^*$-subalgebra of $\mathbb G_2$. Whether that compact quantum group is a compact quantum quotient group as well is still unknown \cite{Wang1995}. 
\end{remark}

\subsection{Inducing monoidal equivalences on supergroups}\label{subsec-inducingonsupergroups}
In this subsection we describe, given a monoidal equivalence $\varphi: \mathbb G_1\to \mathbb G_2$, how to construct a monoidal equivalence between certain quantum supergroups of $\mathbb G_1$ and $\mathbb G_2$. \\Ê
So, let $\mathbb G_1$ and $\mathbb G_2$ be two compact quantum groups and let $\varphi: \mathbb G_1\to \mathbb G_2$ be a monoidal equivalence. Moreover suppose $\mathbb G_1$ is a compact quantum subgroup of a compact quantum group $\mathbb H_1$. As we have done in subsection \ref{subsec-inducing} for Woronowicz $C^*$-subalgebras, we will describe a method to construct a unitary fiber functor on $\mathbb H_1$ from the monoidal equivalence $\varphi$. \\\\
Let $\pi: C_u(\mathbb H_1)\to C_u(\mathbb G_1)$ be the surjective morphism which is compatible with the quantum group structure. Now note that for a representation $U$ of $\mathbb H$ on a Hilbert space $\Hi$, $(\id_{\Hi}\otimes \pi)U$ is a representation of $\mathbb G_1$. Therefore, for $x\in \Irred(\mathbb H_1)$ define $x_{\mathbb G_1}$ to be the equivalence class of $(\id\otimes \pi)U^x$ as representation of $\mathbb G_1$ and

\begin{itemize}
\item if $(\id\otimes \pi)U^x$ is irreducible, let $\Hi_{x_{\mathbb G_1}}=\Hi_x$;
\item If $(\id\otimes \pi)U^x$ is reducible, say $(\id\otimes \pi)U^x=\oplus_{i=1}^n U^{y_i}$, $y_i\in \Irred(\mathbb G_1)$, then let $\Hi_{x_{\mathbb G_1}}=\oplus_{i=1}^n \Hi_{y_i}$.
\end{itemize}
If $x^1, \ldots, x^r,y^1, \ldots,y^s$ are classes of irreducible representations of $\mathbb H_1$ with $U^{x^i_{\mathbb G_1}}=\oplus_{j_i} U^{z^i_{j_i}}$ and $U^{y^i_{\mathbb G_1}}=\oplus_{k_i} U^{t^i_{k_i}}$, we denote for a morphism $S\in \Mor(x_1\otimes \ldots\otimes x_r,y_1\otimes \ldots,y_s)$, $S_{\mathbb G_1}=\bigoplus_{j_1,\ldots,j_r,k_1,\ldots k_s} S^{j_1,\ldots,j_r}_{k_1,\ldots k_s}$ to be the morphism $S$ but seen as element of $\bigoplus_{j_1,\ldots,j_r,k_1,\ldots k_s} \Mor(z^1_{k_1}\otimes \ldots\otimes z^s_{k_s},t^1_{j_1}\otimes \ldots\otimes t^r_{j_r})$, i.e. $S^{j_1,\ldots,j_r}_{k_1,\ldots k_s}\in\Mor(z^1_{k_1}\otimes \ldots\otimes z^s_{k_s},t^1_{j_1}\otimes \ldots\otimes t^r_{j_r})$.

Then we can define the following map:
\begin{proposition}\label{inducedunitfinfunctonsupergroup}
Let $\mathbb G_1, \mathbb G_2,\mathbb H_1$ and $\varphi$ be as above. For $x\in \Irred(\mathbb H_1)$ with $U^{x_{\mathbb G_1}}=(\id\otimes \pi)U^x=\oplus_{i=1}^n U^{y_i}$, $y_i\in \Irred(\mathbb G_1)$ define $\Hi_{\psi'(x)}=\oplus_{i=1}^n\Hi_{\varphi(y_i)}$ and for $S\in \Mor(x_1\otimes \ldots\otimes x_r,y_1\otimes \ldots,y_s)$ with $S_{\mathbb G_1}=\bigoplus_{j_1,\ldots,j_r,k_1,\ldots k_s} S^{j_1,\ldots,j_r}_{k_1,\ldots k_s}$, let $\psi'(S)=\bigoplus_{j_1,\ldots,j_r,k_1,\ldots k_s} \varphi(S^{j_1,\ldots,j_r}_{k_1,\ldots k_s}).$ Then the collection of maps 
$$\Hi_x\mapsto \Hi_{\psi'(x)} \qquad S\in \Mor(x_1\otimes \ldots\otimes x_r,y_1\otimes \ldots,y_s)\mapsto \psi'(S) $$ constitutes a unitary fiber functor $\psi'$ on $\mathbb H_1$.
 \end{proposition}
The proof follows directly by construction of $\Hi_{\psi'}$ and $\psi'(S)$. By theorem \ref{unfibfunctismoneq}, there exists a compact quantum group $\mathbb H_2$ and a monoidal equivalence $\varphi':\mathbb H_1\to \mathbb H_2$. In theorem \ref{thm-explicitsupergroup} we will describe the bi-Galois object associated to $\varphi$ and the compact quantum group $\mathbb H_2$ explicitly.

\begin{theorem}\label{thm-explicitsupergroup}
Let $\mathbb G_1, \mathbb G_2,\mathbb H_1$ be compact quantum groups such that $\mathbb G_1$ is a compact quantum subgroup of $\mathbb H_1$ with surjective morphism $\pi:C_u(\mathbb H_1)\to C_u(\mathbb G_1)$. Let $\varphi:\mathbb G_1\to \mathbb G_2$ be a monoidal equivalence as above and let $\mathbb H_2$ and $\varphi':\mathbb H_1\to \mathbb H_2$ be the compact quantum group and monoidal equivalence induced by $\varphi$ by propositions \ref{inducedunitfinfunctonsupergroup} and \ref{unfibfunctismoneq}. Denoting by $\B$ the ($\mathbb G_1$-$\mathbb G_2$)-bi-Galois object associated to $\varphi$, by $\tilde\B$ the ($\mathbb G_2$-$\mathbb G_1$)-bi-Galois object associated to $\varphi^{-1}$ and by $\B'$ the ($\mathbb H_1$-$\mathbb H_2$)-bi-Galois object associated to $\varphi'$, we have
\begin{equation}\label{bigaloisobjectsupergroup}
\B'\cong \OHa\boksdot[\OGa]\B
\end{equation}
and
\begin{equation}\label{equivsupergroup}
\OHb\cong \tilde \B\boksdot[\OGa]\OHa\boksdot[\OGa]\B
\end{equation}
using the right resp.\ left coactions $(\id\otimes \pi)\Delta_{\mathbb H_1}:\OHa\to \OHa\odot \OGa$ resp.\ $(\pi\otimes \id)\Delta_{\mathbb H_1}:\OHa\to \OGa\odot \OHa$ of $\OGa$ on $\OHa$. 
\end{theorem}
\begin{proof}
Let $X^x$, $x\in \Irred(\mathbb G_1)$ be the elements from theorem \ref{monoidalGalois} associated to $\varphi$. Define for $x\in \IrredHa$, $X^{x_{\mathbb G_1}}=\oplus_{i=1}^n X^{y_i}$ if $U^{x_{\mathbb G_1}}=(\id\otimes \pi)U^x=\oplus_{i=1}^n U^{y_i}$, $y_i\in \Irred(\mathbb G_1)$. Moreover define for $x\in \IrredHa$,
\begin{equation}
Y^x=U^x_{12}X^{x_{\mathbb G_1}}_{13}\in B(\Hi_{\varphi'(x)},\Hi_x)\odot  \OHa\odot\B. 
\end{equation}

We claim that the $Y^x$ with the functional $\omega'=h_{\mathbb H_1}\otimes \omega$ ($h_{\mathbb H_1}$ is the Haar state of $\mathbb H_1$) satisfy the properties 1(a), 1(b) and 1(c) of theorem \ref{monoidalGalois} applied to $\varphi'$.
Indeed, we have for $x,y,z\in \IrredHa$ and $S\in \Mor(y\otimes z,x)$
\begin{eqnarray*}
Y^y_{13}Y^z_{23}(\varphi'(S)\otimes \id)&=&U^y_{13}X^{y_{\mathbb G_1}}_{14}U^z_{23}X^{z_{\mathbb G_1}}_{24}(\varphi'(S)\otimes \id)\\
&=&U^y_{13}U^z_{23}X^{y_{\mathbb G_1}}_{14}X^{z_{\mathbb G_1}}_{24}(\varphi'(S)\otimes \id)\\
&=&U^y_{13}U^z_{23}(S\otimes \id)X^{x_{\mathbb G_1}}_{13}\\
&=&(S\otimes \id)U^x_{12}X^{x_{\mathbb G_1}}_{13}\\
&=&(S\otimes \id)Y^x.
\end{eqnarray*}
Moreover $(\id\otimes\omega')Y^x=(\id\otimes h_{\mathbb H_1}\otimes \omega)(U^x_{12}X^{x_{\mathbb G_1}}_{13})=0$ if $x\neq \varepsilon$.\\\\
Hence to prove \eqref{bigaloisobjectsupergroup} it suffices to prove that the matrix coefficients of the $Y^x$ constitute a linear basis of $\OHa\boksdot[\OGa]\B$. Note first that the matrix coefficients of the $Y^x$ are elements of $\OHa\boksdot[\OGa]\B$. Indeed, $$\big(\id\otimes(\id_\OHa\otimes \pi)\Delta_{\mathbb H_1}\otimes \id_\B\big)U^x_{12}X^{x_{\mathbb G_1}}_{13}=U^x_{12}U^{x_{\mathbb G_1}}_{13}X^{x_{\mathbb G_1}}_{14}=(\id\otimes\id_{\OHa}\otimes \beta_1)U^x_{12}X^{x_{\mathbb G_1}}_{13}.$$ Moreover, as every irreducible representation of $\mathbb G_1$ is a subrepresentation of some $x_{\mathbb G_1}$, $x\in \Irred(\mathbb H_1)$, the matrix coefficients of the $X^{x_{\mathbb G_1}}$ resp.\ the $U^x$ form a basis of $\B$ resp.\ $\OHa$. Hence, the matrix coefficients of the $Y^x$ are linearly independent. Finally we prove that they are also generating. Let $z$ be an arbitrary element of $\OHa\boksdot[\OGa] \B$. Then $z$ is of the form $\sum \lambda^{ij}_{st}u_{ij}^x\otimes b_{st}^y$ where the $u^x_{ij}$ resp.\ $b^y_{st}$ are the matrix coefficients of the $U^x$ resp.\ $X^y$, $x\in \IrredHa,y\in \IrredGa$ and $\lambda^{ij}_{st}\in \mathbb C$. As $z\in \OHa\boksdot[\OGa] \B$, $\sum \lambda^{ij}_{st} u_{ik}^x\otimes \pi(u_{kj}^x)\otimes b_{st}^y=\sum \lambda^{ij}_{st}u_{ij}^x\otimes u^y_{sr} \otimes b_{rt}^y$ and hence $z$ is a linear combination of matrix coefficients of $U^x_{12}X^{x_{\mathbb G_1}}_{13}$. As the unitaries satisfying properties 1(a), 1(b) and 1(c) of theorem \ref{monoidalGalois} are unique, the $Y^x$ are those unitaries and $\B'\cong \OHa\boksdot[\OGa]\B$. This concludes the proof of the first result \eqref{bigaloisobjectsupergroup}.\\Ê\\
For the second result \ref{equivsupergroup}, let $Z^y$, $y\in \Irred(\mathbb G_2)$ be the unitaries from theorem \ref{monoidalGalois} associated to $\varphi^{-1}$. If $U^{x_{\mathbb G_1}}=(\id\otimes \pi)U^x=\oplus_{i}U^{y_i}$ for $x\in \Irred(\mathbb H_1), y_i\in \Irred(\mathbb G_1)$, we will denote $U^{\varphi(x_{\mathbb G_1})}=\oplus_{i}U^{\varphi(y_i)}$ and $Z^{\varphi(x_{\mathbb G_1})}=\oplus_{i}Z^{\varphi(y_i)} \in B(\Hi_x,\Hi_{\varphi'(x)})\odot \tilde \B$.

Therefore, we can define $$V^{\varphi'(x)}=Z^{\varphi(x_{\mathbb G_1})}_{12}U^x_{13}X^{x_{\mathbb G_1}}_{14}.$$
Then, one can prove analogously as above that for $x,y,z\in \IrredHa$ and $S\in \Mor(y\otimes z,x)$
\begin{eqnarray*}
V^{\varphi'(y)}_{13}V^{\varphi'(z)}_{23}(\varphi'(S)\otimes \id)&=&Z^{\varphi(y_{\mathbb G_1})}_{13}U^y_{14}X^{y_{\mathbb G_1}}_{15}Z^{\varphi(z_{\mathbb G_1})}_{23}U^z_{24}X^{z_{\mathbb G_1}}_{25}(\varphi'(S)\otimes \id)\\
&=&Z^{\varphi(y_{\mathbb G_1})}_{13}Z^{\varphi(z_{\mathbb G_1})}_{23}U^y_{14}U^z_{24}X^{y_{\mathbb G_1}}_{15}X^{z_{\mathbb G_1}}_{25}(\varphi'(S)\otimes \id)\\
&=&Z^{\varphi(y_{\mathbb G_1})}_{13}Z^{\varphi(z_{\mathbb G_1})}_{23}U^y_{14}U^z_{24}(S\otimes \id)X^{x_{\mathbb G_1}}_{14}\\
&=&Z^{\varphi(y_{\mathbb G_1})}_{13}Z^{\varphi(z_{\mathbb G_1})}_{23}(S\otimes \id)U^x_{13}X^{x_{\mathbb G_1}}_{14}\\
&=&(\varphi'(S)\otimes \id)Z^{\varphi(x_{\mathbb G_1})}_{12}U^x_{13}X^{x_{\mathbb G_1}}_{14}\\
&=&(\varphi'(S)\otimes \id)V^{\varphi'(x)}.
\end{eqnarray*}
The argument to prove that the matrix coefficients of $V^{\varphi'(x)}$ form a linear basis of $\Hb$ is the same as in the first part of the proof.

\end{proof}
Moreover, the newly constructed compact quantum group $\mathbb H_2$ is a supergroup of $\mathbb G_2$.
\begin{proposition}\label{prop-defissupergroup}
We have a surjective morphism of compact quantum groups $\pi':C_u(\mathbb H_2)\to C_u(\mathbb G_2)$ such that 
\begin{equation}\label{quotmap}
(\id\otimes \pi')V^{\varphi'(x)}= U^{\varphi(x_{\mathbb G_1})}
\end{equation}
for every $x\in \IrredHa$ implying that $\mathbb G_2$ is a quantum subgroup of $\mathbb H_2$.
\end{proposition}
\begin{proof}
The map $\pi'$ is well defined by \eqref{quotmap} as the matrix coefficients of the $V^{\varphi'(x)}$ constitute a linear basis of $\OHb$. Moreover, it is a linear surjection and it follows directly that it is coalgebra map. It suffices to prove that $\pi'$ is an algebra map. Therefore, denoting by $f: \OGb\to \tilde B\boksdot[\OGa]\B$ the isomorphism of proposition \ref{propGbtildeBB} (applied to $\varphi^{-1}:\mathbb G_2\to \mathbb G_1$) such that $(\id\otimes f)U^{\varphi(x)}=Z_{12}^{\varphi(x)}X^x_{13}$, $x\in \Irred(\mathbb G_1)$ it is easy to see that 
\begin{eqnarray*}
(\id\otimes f^{-1})(\id\otimes \id_{\tilde \B}\otimes \varepsilon \otimes \id_{\B})V^{\varphi'(x)}&=&(\id\otimes f^{-1})(\id\otimes \id_{\tilde \B}\otimes \varepsilon \otimes \id_{\B})(Z^{\varphi(x_{\mathbb G_1})}_{12}U^x_{13}X^{x_{\mathbb G_1}}_{14})\\
&=&(\id\otimes f^{-1})(Z^{\varphi(x_{\mathbb G_1})}_{12}X^{x_{\mathbb G_1}}_{13})\\
&=&\oplus_{i}(\id\otimes f^{-1})(Z^{\varphi(y_i)}_{12}X^{y_i}_{13})\\
&=&\oplus_{i}U^{\varphi(y_i)}=U^{\varphi(x_{\mathbb G_1})}\\
&=&(\id\otimes\pi')V^{\varphi'(x)}
\end{eqnarray*}
if $U^{x_{\mathbb G_1}}=\oplus_{i}U^{y_i}$. Hence $(\id\otimes f^{-1})(\id\otimes \id_{\tilde \B}\otimes \varepsilon \otimes \id_{\B})=(\id\otimes\pi')$ proving that $\pi$ is multiplicative as composition of algebra maps. This concludes the proof.
\end{proof}

Finally we prove that the two monoidal equivalences $\varphi$ and $\varphi'$ make isomorphic deformed spectral triples. 
\begin{proposition}\label{prop-defSTsupergroup}
Let $\mathbb G_1, \mathbb G_2,\mathbb H_1$ be compact quantum groups such that $\mathbb G_1$ is a compact quantum subgroup of $\mathbb H_1$ with surjective morphism $\pi:C_u(\mathbb H_1)\to C_u(\mathbb G_1)$ and let $\varphi:\mathbb G_1\to \mathbb G_2$ be a monoidal equivalence as above. Let $\mathbb H_2$ and $\varphi'$ be the compact quantum group and monoidal equivalence induced by $\varphi$ as in proposition \ref{inducedunitfinfunctonsupergroup}. Suppose $\mathbb H_1$ resp. $\mathbb G_1$ act algebraically and by orientation preserving isometries with a unitary representation $V$ resp. $U$ on a spectral triple $(\A,\Hi,D)$ such that $U=(\id\otimes \pi)V$. Denoting by $\B$ the ($\mathbb G_1$-$\mathbb G_2$)-bi-Galois object associated to $\varphi$, by $\tilde\B$ the ($\mathbb G_2$-$\mathbb G_1$)-bi-Galois object associated to $\varphi^{-1}$ and by $\B'$ the ($\mathbb H_1$-$\mathbb H_2$)-bi-Galois object associated to $\varphi'$, the deformed spectral triples

$$(\A\boksdot[\OGa]\B,\Hi\bokstimes[\Ga]L^2(\B),\tilde D)$$ and $$(\A\boksdot[\OHa]\B',\Hi\bokstimes[\Ha]L^2(\B'), \tilde D')$$ (where $\tilde D'$ is the deformation of $D$ along $\varphi'$) are isomorphic.
\end{proposition}

\begin{proof}
It is easy to see that the map
$$\lambda: \A\boksdot[\OGa]\B\to\A\boksdot[\OHa]\OHa\boksdot[\OGa]\B:z\mapsto (\alpha_V\otimes \id_\B)z_{13}$$ is an isomorphism of $^*$-algebras with inverse $(\id_\A\otimes \varepsilon_{\mathbb H_1}\otimes \id_\B)$.
Moreover, let $\phi:  \Hi\bokstimes[\Ga]L^2(\B)\to\Hi\bokstimes[\Ha]L^2(\OHa\boksdot[\OGa]\B): \eta\to V_{12}\eta_{13}$. Then defining $\phi':  \Hi\bokstimes[\Ha]L^2(\OHa\boksdot[\OGa]\B)\to\Hi\bokstimes[\Ga]L^2(\B) :\xi\mapsto (\id\otimes h_{\mathbb H_1}\otimes \id)V^*_{12}\xi$, one can prove that $\phi(\xi)_{13} =\xi\in \Hi\bokstimes[\Ha]L^2(\OHa\boksdot[\OGa]\B)$ as in proposition \ref{propdefhilbertspace}. Hence, it follows that $\phi'=\phi^{-1}$. Moreover, $\phi\tilde D =  \tilde D'\phi$. Finally, we have for $z\in \A\boksdot[\OGa]\B$ and $\eta \in\Hi\bokstimes[\Ga]L^2(\B)$,
\begin{eqnarray*}
\lambda(z)\phi(\eta)&=&V_{12}z_{13}V_{12}^*V_{12}\eta_{13}\\
&=&V_{12}z_{13}\eta_{13}\\
&=&\phi(z\eta)
\end{eqnarray*}
completing the proof.
\end{proof}
\subsection{Deformation of the quantum isometry group}\label{subse-defQISO}

\subsubsection{Deformation of the universal object in $\mathcal Q_R(\A,\Hi,D) $}\label{subsubsec-defQISO0}
In paragrah subsection, we will investigate how the universal objects in the category $\mathcal Q_R(\A,\Hi,D)$ behave with respect to our deformation procedure. 

\begin{proposition}\label{prop-tildeR}
Let $R$ be a positive invertible operator such that $(\A,\Hi,D)$ is a $R$-twisted spectral triple. Suppose $\mathbb G_1$ be a compact quantum group acting algebraically and by orientation-preserving isometries on $(\A,\Hi,D)$ with a representation $U$ and suppose $\varphi: \mathbb G_1\to \mathbb G_2$ is a monoidal equivalence. Denote by $(\tilde \A, \tilde \Hi, \tilde D)$ the deformed spectral triple (theorem \ref{mainthm}). Then there exists a positive invertible operator $\tilde R$ such that $(\tilde \A,\tilde \Hi,\tilde D)$ is an $\tilde R$-twisted spectral triple on which $\mathbb G_2$ acts by $\tilde R$-twisted volume- and orientation-preserving isometries. Moreover, applying the same construction to $\varphi^{-1}$, we obtain $R$ again.
\end{proposition}
\begin{proof}
We can decompose $\Hi$ as $$\Hi = \oplus_{x\in \Irred(\mathbb G_1)}\Hi_x\otimes W_x$$ for some Hilbert spaces $W_x$ where the direct sum is taken over all $x\in \Irred(\mathbb G_1)$, all with multiplicity one. As $D$ commutes with the representation $U$, $D$ is of the form $D=\oplus_{x\in \Irred(\mathbb G_1)}\id_{\Hi_x}\otimes D_x$ where the $D_x$ are operators $W_x\to W_x$. As $\mathbb G_1$ acts by $R$-twisted volume-preserving isometries, 
$$(\tau_R\otimes \id)(\alpha_U(x))=\tau_R(x)1_{\G}$$ 
for all $x\in \mathcal E_D$, where $\tau_R(x)=\Tr(Rx)$ and where $\mathcal E_D$ is the $^*$-subalgebra of $B(\Hi)$ generated by the rank-one operators of the form $\eta\xi^*$, $\eta,\xi$ eigenvectors of $D$. Therefore, also $(\tau_R\otimes h_{\mathbb G_1})(\alpha_U(x))=\tau_R(x)$ from which it follows (as in the proof of theorem 3.8 of \cite{Goswami2014}) that $R$ must be of the form $R= \oplus_{x\in \Irred(\mathbb G_1)}F_x\otimes R_x$, where $F_x$ is the matrix such that $h_{\mathbb G_1}(u^x_{ij}(u^y_{st})^*)=\frac{\delta_{x,y}\delta_{i,s}(F_x)_{jt}}{\Tr(F_x)}$ (described by Woronowicz \cite{Woronowicz1998}) and $R_x:W_x\to W_x$ positive operators. As $(\A,\Hi,D)$ is an $R$-twisted spectral triple, $R$ and $D$ commute and hence each $D_x$ commutes with $R_x$ for all $x\in \Irred(\mathbb G_1)$. Now, in this presentation $\tilde \Hi=  \oplus_{x\in \Irred(\mathbb G_1)}\Hi_{\varphi(x)}\otimes W_x$ and $\tilde D= \oplus_{x\in \Irred(\mathbb G_1)}\id_{\Hi_{\varphi(x)}}\otimes D_x$. Therefore, define $\tilde R=  \oplus_{x\in \Irred(\mathbb G_1)}F_{\varphi(x)}\otimes R_x$. Then $\tilde R$ is again positive, and invertible and it commutes with $\tilde D$. Moreover, $\mathbb G_2$ acts by $\tilde R$-twisted volume preserving isometries by the defining property of $F_{\varphi(x)}$. It is clear that the inverse construction gives $R$ again.
\end{proof}

\begin{theorem}\label{thm-defisQISO0}
Let $R$ be a positive invertible operator on a Hilbert space $\Hi$ and let $(\A,\Hi,D)$ be an $R$-twisted compact spectral triple on which $\QISORnul$ acts algebraically. Suppose $\varphi: \QISORnul\to \mathbb G_2$ is a monoidal equivalence with bi-Galois object $\B$. Then 
$\mathbb G_2\cong \QISORtildenul$ for $\tilde R$ as in proposition \ref{prop-tildeR}.
\end{theorem}
\begin{remark}\label{rem-algactQISO}
Note that the condition that $\QISORnul$ acts algebraically on $(\A,\Hi,D)$ is not essential. If $\QISORnul$ does not act algebraically on $(\A,\Hi,D)$, we know from proposition \ref{prop-existsalgactsubalg} that there exists a $^*$-algebra $\A_1$ which is SOT-dense in $\A''$ such that $(\A_1,\Hi,D)$ is a compact spectral triple on which $ \QISORnul$ acts algebraically. Moreover, $\QISORnul\cong \QISORnul[\A_1,\Hi,D]$ by proposition 3.9 of \cite{Goswami2014}.
\end{remark}

\begin{proof}[Proof of theorem \ref{thm-defisQISO0}]
By proposition \ref{QISO0existence}, there exists a universal object $\QISORnul$ in the category $\mathcal Q_R$ of compact quantum groups acting by $R$-twisted volume- and orientation preserving isometries on $(\A,\Hi,D)$.  For notational convenience, we will denote this quantum group by $\QISORnulb$. Now, as $\varphi:\QISORnulb \to \mathbb G_2$ is a monoidal equivalence, $\mathbb G_2$ acts algebraically and by orientation preserving isometries on $(\tilde\A,\tilde \Hi,\tilde D)=(\A\boksdot[\mathcal O(\QISORnulb)]\B,\Hi\bokstimes[C(\QISORnulb)]L^2(\B),\tilde D)$. Denote by $\tilde R$ the operator constructed in proposition \ref{prop-tildeR}. Then $\mathbb G_2$ is a quantum subgroup of $\QISORtildenul$. Moreover, the monoidal equivalence $\varphi^{-1}: \mathbb G_2 \to \QISORnul[\A,\Hi,D]$ induces a unitary fiber functor $\psi'$ on $\QISORtildenul$ by proposition \ref{inducedunitfinfunctonsupergroup}; we will denote the deformed quantum group by $\mathbb H_1$ and the monoidal equivalence associated to $\psi'$ (for notational convenience) by $\varphi'^{-1}:\QISORtildenul[\A,\Hi,D]\to \mathbb H_1$ and the associated bi-Galois object by $\tilde \B'$. As $\mathbb G_2$ is a quantum subgroup of $\QISORtildenul$, $\QISORnul$ is a quantum subgroup of $\mathbb H_1$ by proposition \ref{prop-defissupergroup} and both act by $R$-twisted volume- and orientation-preserving isometries on $(\A,\Hi,D)$ by proposition \ref{prop-defSTsupergroup}. Hence by universality
%hence it is a subgroup of $\QISORnuleen$. We have the following inclusions of quantum groups
%$$\QISORnuleen \leq \wQISORtildenuleen\leq \QISORnuleen$$ implying 
\begin{equation}
\QISORnuleen\cong  \mathbb H_1. 
\end{equation}
and also
$$\mathbb G_2\cong  \QISORtildenul.$$
This completes the proof.
\end{proof}

\subsubsection{Deformation of the quantum isometry group}\label{subsubsec-defQISO}
In this paragraph we use subsection \ref{subsec-inducing} and paragraph \ref{subsubsec-defQISO0} to strengthen the result of theorem \ref{thm-defisQISO0} to quantum isometry groups.
\begin{theorem}\label{thm-defisQISO} 
Let $(\A,\Hi,D)$ be an $R$-twisted compact spectral triple such that $\QISORnul$ acts algebraically on $(\A,\Hi,D)$. Suppose moreover that we have a monoidal equivalence $$\varphi:\QISORnul\to \mathbb G_2.$$ Then there exists a monoidal equivalence 
$$\varphi': \QISOR\to\QISORtilde$$
where $(\tilde \A,\tilde \Hi,\tilde D)$ is the spectral triple obtained by deformation with $\varphi$ by theorem \ref{mainthm} and $\tilde R$ the operator obtained from proposition \ref{prop-tildeR}.
\end{theorem}
\begin{remark}
One can make again remark \ref{rem-algactQISO} here.
\end{remark}

\begin{proof}[Proof of theorem \ref{thm-defisQISO}]
Denote the universal object of $\Q_R$ for notational convenience by $\QISORnulb=(C(\QISORnulb),U_0)$. Analogously $\QISOb^0_{\tilde R}=\QISORtildenul$. As $C(\QISOb_R)=C^*(\{(f\otimes \id)\alpha_U(a)\mid a\in \A,f\in \A^*\})$, it is a Woronowicz $C^*$-subalgebra of $\QISORnulb$ and hence we can apply the theory of section \ref{subsec-inducing}. We obtain a compact quantum group $\mathbb H_2$ and a monoidal equivalence $\varphi':\QISOb_R\to\mathbb H_2$ and it suffices to prove $\mathbb H_2=\QISORtilde$. Note now that as $\QISORnulb$ acts algebraically on $(\A,\Hi,D)$, we can decompose $\A$ into spectral subspaces $\A_x$ and define the subset $I$ of $\Irred(\QISORnulb)$ by $I=\{x\in \Irred(\QISORnulb)\mid\A_x\neq 0\}$. Then we have $C(\QISOb_R)=C^*(\{u_{ij}^x\mid x\in I \})$ and $I=\Irred(\QISOb_R)$. Moreover, $C(\mathbb H_2)=C^*(\{u_{ij}^{\varphi(x)}\mid x\in I\})$ and by theorem 7.3 of \cite{DeRijdt2010}, we know that also $I= \{x\in \Irred(\QISORnulb)\mid\tilde \A_{\varphi(x)}\neq 0\}$. Hence we can conclude that $\mathbb H_2=\QISORtilde$.

This concludes the proof.

\end{proof}

\subsection{Deformation of the quantum isometry group of the Podle\'s sphere}
In this last subsection, we use subsection \ref{subse-defQISO} to find the quantum isometry group of the newly constructed spectral triple in theorem \ref{thm-defpodles}. Therefore we investigate first the quantum isometry group of the Podle\'s sphere. 

\begin{definition}[\cite{Podles1987}]
Define $B$ to be the unital $^*$-subalgebra of $C(SU_q(2))$ generated (as $^*$-algebra) by the elements $\alpha^2, \gamma^*\gamma,\gamma^2,\alpha\gamma$ and $\gamma^*\alpha$. The closure of $B$ is a Woronowicz $C^*$-algebra of $SU_q(2)$ and the associated compact quantum group is called $SO_q(3)$. 
\end{definition}
%\begin{remark}\label{basisSOq}
%Using proposition \ref{basisSUq}, we see that the elements of the form $\alpha^i\gamma^j(\gamma^k)^*$,$i,j,k\in \mathbb{N}$ with $i+j+k$ even and $(\alpha^*)^l\gamma^m(\gamma^*)^n$, $m,n\in \mathbb{N},l\in\mathbb{N}\setminus \{0\}$ with $l+m+n$ even form a basis of $SO_q(3)$.
%\end{remark}
In the classical situation, we know that $SO(3)$ is a quotient group of $SU(2)$, indeed $SO(3)=SU(2)/\{-1,1\}$. In the quantum versions this is also true: we can prove that $\mathbb{Z}_2$ is a normal quantum subgroup of $SU_q(2)$ and $SU_q(2)/\mathbb{Z}_2$ equals $SO_q(3)$.\\

\begin{theorem}[\cite{Bhowmick2010b}]
Let $S^2_{q,c}$ be the Podle\'s sphere as defined in subsection \ref{subsec-mondefpodles}. Then $$\QISOR[\mathcal O(S^2_{q,c}),\Hi,D]\cong SO_q(3).$$
\end{theorem}
Now we will investigate monoidal equivalences of $SO_q(3)$ in order to apply theorem \ref{thm-defisQISO} to find the quantum isometry group of the spectral triples constructed in \ref{thm-defpodles}.\\
%\subsubsection{Monoidal equivalences on $SO_q(3)$}
We defined $SO_q(3)$ as coming from a Woronowicz-$C^*$-subalgebra of $SU_q(2)$. Using the theorems of subsection \ref{subsec-inducing}, we will use the induction method to construct monoidal equivalences on $SO_q(3)$. 
%Now we have prepared all ingredients to apply the construction of paragraph \ref{subsec-inducing} to our concrete example. 
Therefore fix a monoidal equivalence between $SU_q(2)$ and a suitable $A_o(F')$ with $\dim(F')\geq 3$. %, then from proposition \ref{clasmoneqSUq}, we see that $q\leq 2-\sqrt 3$. 
As $SO_q(3)=SU_q(2)/\mathbb{Z}_2$, we find a Woronowicz subalgebra $R(F')$ of $A_o(F')$ such that $SO_q(3)$ is monoidally equivalent with $R(F')$. Now Theorem 4.1 in \cite{Wang2009}, gives us a concrete description of $R(F')$.

\begin{theorem}[Theorem 4.1 in \cite{Wang2009}]
Let $F\in \GL(n,\mathbb{C})$ be such that $F\overline{F}=\pm I_n$. Then every Woronowicz subalgebra of $A_o(F)$ is a quantum quotient group. Moreover it has only one normal subgroup of order 2 with quantum quotient group $C^*(r_{2m})$ (where $r_{2m}$ is the irreducible representation of dimension $2m$).
\end{theorem}
Applying this theorem to $F=F_q$, it affirms that $SO_q(3)$ is the only compact quantum quotient group of $SU_q(2)$. Applying it to $F=F'$, we get a concrete description of $R(F')$. By remark \ref{nondimpreserving}, it can be seen that the induced monoidal equivalence is not dimension-preserving and hence not a 2-cocycle deformation (by proposition \ref{cocycle-unfibfunct}).

Combining all of this, we get

\begin{theorem}\label{mainthmcha5}
Let $F\in \GL(n,\mathbb{C})$ be such that $F\overline{F}=\pm I_n$ and $\varphi:SU_q(2)\to A_o(F)$ a monoidal equivalence with bi-Galois object $\B=A_o(F_q,F)$. Define $I(F)$ to be the $C^*$-algebra generated by the 
$U_{ij}U_{kl}$ where $U$ is the unitary in $M_n(A_o(F))$ satisfying the relation $U=F\overline U F^{-1}$ as in definition \ref{defAoF}. Define $P(F_q,F)$ to be the $^*$-algebra generated by the $Y_{ij}Y_{kl}$ where $Y$ is the unitary in $M_{n_2,n_1}(\mathbb{C})\otimes C_u(A_o(F_1,F_2))$ described in theorem \ref{moneqAoF}.
Then there exists a monoidal equivalence $\varphi':SO_q(3)\to I(F)$ with bi-Galois object $\B'=P(F_q,F)$ which is not dimension-preserving (by remark \ref{nondimpreserving}).
\end{theorem}
Now we are ready to characterize the quantum isometry groups of the spectral triples constructed in \ref{thm-defpodles}.

\begin{theorem}\label{thm-defpodleswithQISO}
Let $q\in [-1,1]\setminus\{0\}$ and $n$ a natural number with $3\leq n\leq |q+1/q|$. If $q>0$, suppose $n$ is even.
With the matrix $F$ defined as in theorem \ref{thm-defpodles}, $I(F)$ as constructed in theorem \ref{mainthmcha5} is the quantum isometry group of the spectral triple $$\big(\mathcal{O}(S_{q,c}^2)\boksdot[\mathcal{O}(SU_q(2))] \mathcal O(A_o(F_q,F)),\quad\Hi\bokstimes[C(SU_q(2))] L^2\big(\mathcal O(A_o(F_q,F))\big), \quad\tilde D\big)$$ from theorem \ref{thm-defpodles}. 
\end{theorem}

%\subsubsection{Deformation of the quantum isometry group of the Podle\'s sphere}

%After that, we will recall the result in \cite{Bhowmick2010b} that $SO_q(3)$ is the quantum isometry group of $S^{2}_{q,c}$ and finally, we will use the unitary fiber functors on $SO_q(3)$ developed above to find a new spectral triple with a quantum group acting on it in an isometrical way.
%
%Furthermore, in \cite{Bhowmick2010b}, Bhowmick and Goswami proved that $SO_q(3)$ acts isometrically and by orientation-preserving isometries on $(C(S^2_{q,c}),\Hi, D)$. Indeed, from the description of $S^2_{q,c}$ with $\tilde A$ and $\tilde B$ as in the last section, one can prove that the coproduct of $SU_q(2)$ restricted to the Podle\'s sphere is indeed an action $C(S^2_{q,c})\to C(S^2_{q,c})\otimes C(SO_q(3))$. Even more, the action is algebraic and $SO_q(3)$ is the quantum isometry group (defined in \cite{Bhowmick2009}) of $S^2_{q,c}$.

\section*{Acknowledgments}
I want to thank Johan Quaegebeur and Pierre Bieliavsky for many fruitful discussions and most valuable and appreciated advice. Also Kenny De Commer, Yuki Arano and Stefaan Vaes are warmly thanked for discussing about the subject. The valuable and detailed comments and suggestions of the referees are warmly appreciated. They made it possible to improve the paper substantially. Finally I want to thank the IAP DYGEST for financial support.

\bibliographystyle{acm}
\bibliography{paperldesadeleerrevision}

%Where the bibliography will be printed
 %\printbibliography

\end{document}